\documentclass[12pt,draftclsnofoot,onecolumn]{IEEEtran}

\makeatletter
\newcommand\semihuge{\@setfontsize\semihuge{22.3}{22}}
\makeatother

\usepackage{algpseudocode}
\usepackage{algorithm}
\usepackage{algorithmicx}

\usepackage{lipsum} 

\usepackage[dvips]{color}
\usepackage{comment}
\usepackage{todonotes}
\usepackage{epsf}
\usepackage{epsfig}
\usepackage{times}
\usepackage{epsfig}
\usepackage{graphicx}
\usepackage{bbold}
\usepackage{mathtools}
\usepackage{mathrsfs}
\usepackage{amssymb}
\usepackage{pdfpages}
\usepackage{epstopdf}
\usepackage{float}
\newfloat{algorithm}{t}{lop}

\usepackage{dsfont}
\usepackage{lettrine} 
\usepackage{amsmath,epsfig,amssymb,algorithm,algpseudocode,amsthm,cite,url}
\usepackage{subcaption}
\allowdisplaybreaks
\usepackage{csquotes}

\usepackage{verbatim}
\usepackage[english]{babel}
\usepackage{amsmath,amssymb}

\usepackage[justification=centering]{caption}
\usepackage{verbatim}

\newtheorem{theorem}{\bf Theorem}

\newtheorem{proposition}{\bf Proposition}
\newtheorem{lemma}{\bf Lemma}

\begin{document}
	
\title{\semihuge  Wireless Communication using Unmanned Aerial Vehicles (UAVs): Optimal Transport Theory for Hover Time Optimization \vspace{-0.3cm}}    

\author{\IEEEauthorblockN{  Mohammad Mozaffari$^1$, Walid Saad$^1$, Mehdi Bennis$^2$, and M\'erouane Debbah$^3$}\vspace{-0.05cm}\\
	\IEEEauthorblockA{
		\small $^1$ Wireless@VT, Electrical and Computer Engineering Department, Virginia Tech, VA, USA,\\ Emails:\url{{mmozaff , walids}@vt.edu}.\vspace{-0.07cm}\\
		$^2$ CWC - Centre for Wireless Communications, Oulu, Finland, Email: \url{bennis@ee.oulu.fi}.\vspace{-0.07cm}\\
		$^3$ Mathematical and Algorithmic Sciences Lab, Huawei France R \& D, Paris, France, and CentraleSup´elec,\vspace{-0.07cm}\\   Universit´e Paris-Saclay, Gif-sur-Yvette, France, Email: \url{merouane.debbah@huawei.com}.
	}\vspace{-0.92cm}}
\maketitle\vspace{-0.8cm}
\vspace{-0.1cm}
\begin{abstract}\vspace{-0.4cm}	
\textcolor{black}{In this paper, the effective use of flight-time constrained unmanned aerial vehicles (UAVs) as flying base stations that can provide wireless service to ground users is investigated. In particular, a novel framework for optimizing the performance of such UAV-based wireless systems in terms of the average number of bits (data service) transmitted to users as well as UAVs' hover duration (i.e. flight time) is proposed.} In the considered model, UAVs hover over a given geographical area to serve ground users that are distributed within the area based on an arbitrary spatial distribution function. In this case, two practical scenarios are considered. In the first scenario, based on the maximum possible hover times of UAVs, the average data service delivered to the users under a fair resource allocation scheme is maximized by finding the optimal cell partitions associated to the UAVs. \textcolor{black}{Using the powerful mathematical framework of optimal transport theory, this cell partitioning problem is proved to be equivalent to a convex optimization problem. Subsequently, a gradient-based algorithm is proposed} for optimally partitioning the geographical area based on the users' distribution, hover times, and locations of the UAVs.  In the second scenario, given the load requirements of ground users, the minimum average hover time that the UAVs need for completely servicing their ground users is derived. To this end, first, an optimal bandwidth allocation scheme for serving the users is proposed. Then, given this optimal bandwidth allocation, the optimal cell partitions associated with the UAVs are derived by exploiting the optimal transport theory. Simulation results show that our proposed cell partitioning approach leads to a significantly higher fairness among the users compared to the classical weighted Voronoi diagram
. Furthermore, the results demonstrate that the average hover time of the UAVs can be reduced by 64\% by adopting the proposed optimal bandwidth allocation as well as the optimal cell partitioning approach. In addition, our results reveal an inherent tradeoff between the hover time of UAVs and bandwidth efficiency while serving the ground users.

\end{abstract} \vspace{0.1cm}

\section{Introduction}

Recently, the use of aerial platforms such as unmanned aerial vehicles (UAVs), drones, balloons, and helikite has emerged as a promising solution for providing reliable and cost-effective wireless communication services for ground wireless devices \cite{orfanus,mozaffari2, Ismail, Zhang}. In particular, UAVs can be deployed as flying base stations for coverage expansion and capacity enhancement of terrestrial cellular networks \cite{mozaffari2,Ismail, IoTJournal,HouraniModeling, Azari, Kalantari, Zhang,  bor,Letter}. With their inherent attributes such as mobility, flexibility, and adaptive altitude, UAVs have several key potential applications in wireless systems. 
 For instance, UAVs can be deployed to complement existing cellular systems by providing additional capacity to hotspot areas during temporary events. Moreover, UAVs can also be used to provide network coverage in emergency and public safety situations during which the existing terrestrial network is damaged or not fully operational. One key advantage of UAV-based wireless communication is its unique ability to provide fast, reliable and cost-effective connectivity to areas which are poorly covered by terrestrial networks. In addition, compared to ground base stations, UAVs can more effectively establish line-of-sight (LoS) communication links to ground users by intelligently adjusting their altitude. Key examples of recent projects on employing aerial platforms for wireless connectivity include Google Loon project and Facebook's Internet-delivery drone \cite{Sky}. Within the scope of these practical deployments, UAVs are being used to deliver Internet access to emerging countries and provide airborne global Internet connectivity. Despite the several benefits and practical applications of using UAVs as aerial base stations, one must address many technical challenges such as performance analysis, deployment, air-to-ground channel modeling, user association, and flight time optimization \cite{Zhang} and \cite{bor}.\vspace{-0.3cm}
 
\subsection{Related Works}

In \cite{HouraniModeling}, the authors performed air-to-ground channel modeling for UAV-based communications in various propagation environments. 
 In \cite{Azari} and \cite{Kalantari}, the authors studied the efficient deployment of aerial base stations to maximize the coverage and rate performance of wireless networks.  In \cite{Jeong}, the authors investigated the  energy-efficient path planning of a UAV-mounted cloudlet which is used to provide offloading opportunities to ground devices. The work in \cite{jiang} studied the optimal trajectory and heading of UAVs for sum-rate maximization in uplink communications. An analytical framework for trajectory optimization of a fixed-wing UAV for energy-efficient communications was presented in \cite{ZhangEnergy}. The work in \cite{Qin} jointly optimized user scheduling and UAV trajectory
for maximizing the minimum average rate among ground users. The authors in \cite{Vishnu} derived an exact expression for downlink coverage probability for ground receivers which are served by multiple UAVs. 

Another important challenge in UAV-based communications is user (or cell) association. The work in \cite{Vishal} investigated the area-to-UAV assignment for capacity enhancement of heterogeneous wireless networks. However, this work is limited to the case with a uniform spatial distribution of ground users. Moreover, the work in \cite{Vishal} does not consider any fairness criteria that can be affected by the network congestion in a non-uniform users' distribution case. In addition, this work ignores the UAVs' flight time constraints while determining the cell partitions associated with the UAVs. In \cite{OTUAV} the optimal cell partitions associated to UAVs were determined with the goal of minimizing the UAVs' transmit power while satisfying the users' rate requirements. However, in \cite{OTUAV}, the impact of flight time constraints on the performance of UAV-to-ground communications is not taken into account.    

Indeed, the \emph{flight time duration of the UAVs} presents a unique design challenge for UAV-based communication systems \cite{niu} and \cite{DroneDel}.
For instance, the performance of such systems significantly depends on the \emph{hover time} of each UAV, which is defined as the flight time during which the UAV must stay in the air over a given area for providing wireless service to ground users. In fact, with a higher hover time of the UAV, the users can receive wireless service for a longer period. Thus, by increasing the hover time, a UAV can meet higher load requirements and serve a larger area. However, the hover time of a UAV is naturally limited due to the highly constrained battery-provided, on-board energy, as well as flight regulations such as no-fly time/zone constraints \cite{AkramMagazin}. Hence, while analyzing the UAV-based communication systems, the hover time constraints must be also taken into account. In this case, there is a need for a framework to analyze and optimize the performance of UAV-based communications based on the hover time of UAVs. In fact, to our best knowledge, none of the previous UAV studies such as \cite{IoTJournal, Letter, Ismail,Zhang, mozaffari2, ZhangEnergy, HouraniModeling, bor, Kalantari, Vishal, Vishnu,Jeong,jiang,Sky,orfanus, Azari,Qin}, considered the hover time constraints in their analysis. \vspace{-0.3cm} 

\subsection{Contributions}
The main contribution of this paper is to develop a novel framework for  optimized UAV-to-ground communications under explicit UAVs' hover time constraints. In particular, we consider a network in which multiple UAVs are deployed as aerial base stations to provide wireless service to ground users that are distributed over a geographical area based on an arbitrary spatial distribution. We investigate two key practical scenarios: \emph{UAV communication under hover time constraints}, and \emph{UAV communication under load constraints}. In the first scenario, given the maximum possible hover time of UAVs that is imposed by the limited on-board energy of UAVs and flight regulations, we maximize the average number of bits (data service) that is transmitted to the users under a fair resource allocation scheme. To this end, given the hover times and the spatial distribution of users, we find the optimal cell partitions associated to the UAVs. In this case, using the powerful mathematical framework of \emph{optimal transport theory} \cite{villani}, we propose a gradient-based algorithm that optimally partitions the geographical area based on the users' distribution as well as the UAVs' hover times and locations. In the second scenario, given the load requirements of ground users, we minimize the average hover time needed for completely serving the users. To this end, we introduce an optimal bandwidth allocation scheme as well as optimal cell partitions for which the average hover time of UAVs is minimized. Our results for the first scenario show 
that our proposed cell partitioning approach leads to a significantly higher fairness among the users compared to the classical weighted Voronoi approach
. For the second scenario, the results show that the average hover time can be reduced by 64\% by adopting the proposed optimal bandwidth allocation and cell partitioning approach. Furthermore, our results reveal an inherent tradeoff between the hover time of UAVs and bandwidth efficiency while servicing the ground users. \vspace{-0.01cm} 

The rest of this paper is organized as follows. In Section II, we present the system model. In Section III we investigate Scenario 1 in order to maximize data service. Section IV presents Scenario 2 for minimizing the hover time of UAVs. Simulation results are presented in Section V and conclusions are drawn in Section VI.

\section{System Model}

Consider a geographical area $\mathcal{D}\subset \mathds{R}^2$ within which a number of wireless \textcolor{black}{users are located according to} a given distribution $f(x,y)$ in the two-dimensional plane. In this area, a set  $\mathcal{M}$ of $M$ UAVs are used as aerial base stations to provide wireless service for the ground users\footnote{For wireless backhauling of aerial networks, satellite and WiFi are considered as the two feasible candidates \cite{AkramMagazin}.}.  Let $\boldsymbol{s}_i=(x_{i},y_{i},h_i)$ be the three-dimensional (3D) coordinate of each UAV $i\in \mathcal{M}$ with $h_i$ being the altitude of UAV $i$. 
We consider a downlink
scenario in which each UAV adopts a frequency division multiple access (FDMA) technique to provide service for the ground users as done in \cite{mozaffari2} and \cite{FDMA2}. 
Let $P_i$ and $B_i$ be, respectively, the maximum transmit power and the total available bandwidth for UAV $i$. 
 Moreover, we use $\mathcal{A}_i$, as shown in Fig.\,\ref{SystemModel},  to denote the partition of the geographical area which is served by UAV $i$. In this case, all users located in cell partition $\mathcal{A}_i$ will be connected to UAV $i$. Hence, the geographical area is divided into $M$ disjoint partitions each of which is associated with one of the UAVs. Let $\tau_i$ be the \emph{hover time} of UAV $i$, defined as the time duration that a UAV uses to hover (stop) over the corresponding cell partition to service the ground users. During the hover time, the UAV must  initiate connections to the ground users, perform required computations, and transmit data to the users. Let $T_i$ be the effective data transmission period during which a UAV services the users. In general, the effective data  transmission time is less than the total hover time. Consequently, we consider a \emph{control time} as ${g_i(.)}$,  a function of the number of users in $\mathcal{A}_i$, to represent the portion of the hover time that is not used for the effective data transmission. This control time naturally captures the total time that a UAV $i$ needs to spend for computations, setting up connections, and control signaling. Intuitively, the control time will increase when the number of users in the corresponding cell partition increases.  
 
  In our model, we use the term \textit{data service} to represent the amount of data (in bits) that each UAV transmits to a given user. Clearly, the data service depends on several factors such as the effective data transmission time (which is directly related to flight time), and the transmission bandwidth. 
  Therefore, here, the effective data transmission times and bandwidth of the UAVs are considered as \textit{resources} which are used for servicing the users.  
    \begin{figure}[!t]
    	\begin{center}
    		\vspace{-0.1cm}
    		\includegraphics[width=8cm]{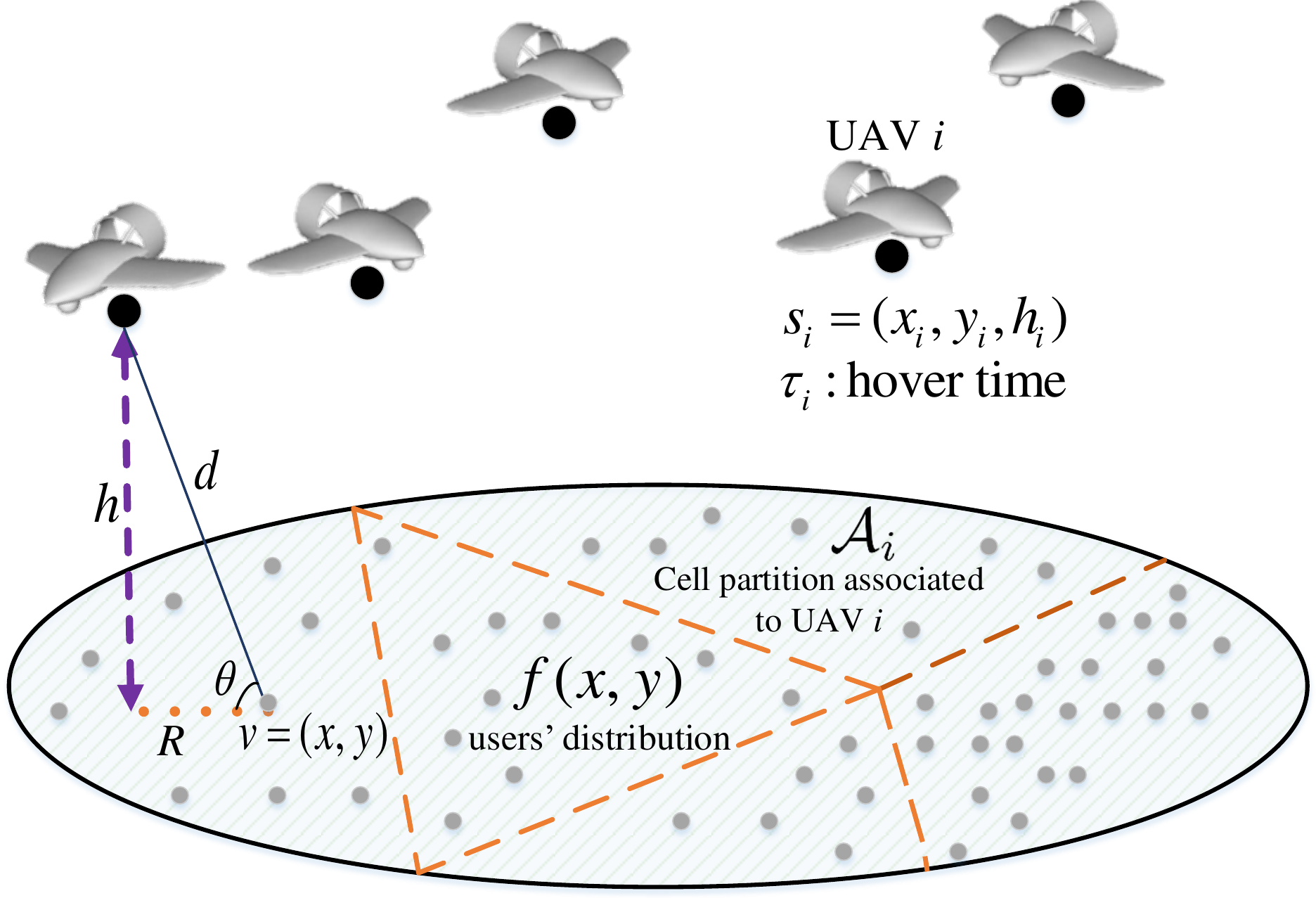}
    		\vspace{-0.4cm}
    		\caption{  System model. \vspace{-.82cm}}
    		\label{SystemModel}
    	\end{center}	\vspace{-0.4cm}
    \end{figure} 
 Given this model, to provide service for the ground users using UAVs, we consider two scenarios. The first scenario, \textbf{Scenario~1}, can be referred to as  \emph{UAV communications under hover time constraints}. 
   In this case, given the maximum possible hover times (imposed by the energy and flight limitations of each UAV), we maximize the average data service to the users under a  fair resource allocation policy by optimal cell partitioning of the area. Here, we optimally partition the geographical area based on the hover times and the spatial distribution of users. In Scenario 1, given the maximum possible hover time of each UAV, the total data service under user fairness considerations is maximized. In fact, Scenario 1 corresponds to resource-limited communication scenarios in which the amount of resources (e.g. hover times and bandwidth) is not sufficient to completely meet the demands. An example of such scenario is when battery-limited UAVs are deployed in hotspots with high number of users and demands. The second scenario, \textbf{Scenario 2}, is referred to as  \emph{UAV communication under load constraints}. In this case, our goal is to completely meet the demands of ground users by properly adjusting the hover time of the UAVs. In particular, given the load requirement of each user at a given location, we minimize the average hover time needed for completely serving the ground users. As a result, the load requirement of the ground users will be satisfied with a minimum average hover time of the UAVs. In this case, by minimizing the hover time, one can minimize the energy consumption of the UAVs as well as the time needed to completely serve the ground users. Such analyses in Scenario 2 are primarily useful is emergency situations in which all users need to be quickly served by the UAVs. 
    \vspace{-0.4cm}

\subsection{Air-to-ground path loss model}
The air-to-ground signal propagation is affected by
the obstacles and buildings in the environment. In this case, depending on the propagation environment,
air-to-ground communication links can be either LoS or non-line-of-sight (NLoS). In general, while designing a UAV-based communication system, a complete information about the exact locations, heights, and
the number of obstacles may not be available \cite{HouraniModeling}. In such case, one must consider the randomness associated with the LoS and NLoS links \cite{Kalantari}, \cite{HouraniModeling}. Clearly, the probability of having LoS communication links depends on the locations, heights, and the number of obstacles, as well as the elevation angle between a given UAV and it's served ground user. In our model, we consider a widely used probabilistic path loss model provided by International Telecommunication Union (ITU-R) \cite{ITUR}, and the work in \cite{HouraniModeling}. In this case, the path loss between UAV $i$ and a given user at location $(x,y)$ can be given by  \cite{Kalantari}, and \cite{HouraniModeling}:\vspace{-0.3cm}

\begin{equation}\label{Pr}
	{\Lambda_{i}(x,y)} = \left\{\hspace{-0.16cm} \begin{array}{l}
		\left(\frac{{4\pi {f_c}d_o}}{c}\right)^{2}\big(d_i(x,y)/d_o\big)^{2} {\mu _\text{LoS},} \,\,\,\hspace{0.35cm}{\text{LoS link,}}\\
		 \left(\frac{{4\pi {f_c}d_o}}{c}\right)^{2}\big(d_i(x,y)/d_o\big)^{2} {\mu _\text{NLoS},} \,\,\,\hspace{0.21cm}{\text{NLoS link,}}
	\end{array} \right.
\end{equation}
where $\mu_\textrm{LoS}$ and $\mu_ \textrm{NLoS}$ are different attenuation factors considered for LoS and NLoS links. 
Here, $f_c$ is the carrier frequency, $c$ is the speed of light, and $d_o$ is the free-space reference distance.  \begin{small}$d_i(x,y)=\sqrt{(x-x _{i})^2+(y-y_{i})^2+h_i^2}$\end{small} is the distance between UAV $i$ and an arbitrary ground user located at $(x,y)$. For the UAV-user link, the LoS probability is given by \cite{HouraniModeling}:\vspace{-0.1cm} 
\begin{equation} \label{PLoS}
	{P_{\text{LoS},i}} = b_1 {\left( {\frac{180}{\pi}\theta_i  - 15} \right)^{b_2} }, 
\end{equation}
 where $\theta_i={\sin ^{- 1}}( \frac {h_i} {d_i(x,y)})$ is the elevation angle (in radians) between the UAV and the user. Also, $b_1$ and $b_2$ are constant values reflecting the environment impact. Note that, the NLoS probability is $P_{\text{NLoS},i}=1-P_{\text{LoS},i}$. Clearly, considering $d_o=1$\,m, and  \begin{small}$K_o=\left(\frac{4\pi f_c} {c}\right)^{2}$\end{small} the average path loss is\\ ${ {{K_o d_i}}^{ 2}(x,y)}[ {{P_{\textrm{LoS},i}}{\mu _\textrm{LoS}} +  {P_{\textrm{NLoS},i}}{\mu _\textrm{NLoS}}} ]$. Hence, the received signal power from UAV $i$ will be:
\begin{equation}\label{Puav_ave}
 \bar P_{r,i}(x,y) =\frac{P_i}{{ K_o{{d_i}}^{ 2}(x,y)}\left[ {{P_{\textrm{LoS},i}}{\mu _\textrm{LoS}} + {P_{\textrm{NLoS},i}}{\mu _\textrm{NLoS}}} \right]},
\end{equation}
where $P_i$ is the UAV's transmit power. Then, the received SINR for a user located at $(x,y)$ while connecting to UAV $i$ will be:
\begin{equation} 
{\gamma _i}(x,y) = \frac{{{{\bar P}_{r,i}}(x,y)}}{{{I_i}(x,y) + {\sigma ^2}}}, \label{gamma}
\end{equation}
where ${I_i}(x,y) = \beta \sum\limits_{j \ne i} {{{\bar P}_{r,j}}(x,y)}$ is the received interference at location $(x,y)$ stemming from all UAVs except UAV $i$. We also consider a weight factor $0\le \beta \le 1$ to adjust the amount of interference and capture the impact of any interference mitigation technique. Naturally, $\beta=1$ and  $\beta=0$, respectively, correspond to the full interference and interference-free scenarios.
 
Clearly, the throughput of a user located at $(x,y)$ if it connects to UAV $i$ is:\vspace{-0.2cm} 
\begin{align} \label{Cap}
&C_i (x,y)= {{B(x,y)} \,{{\log }_2}\left( {1 + \gamma_i(x,y)} \right)}, 
\end{align}
where $B(x,y)$ is the bandwidth allocated to the user at $(x,y)$. 

Subsequently, the total data service for the user provided by the UAV will be:\vspace{-0.2cm} 
\begin{equation} \label{Loi}
{L_i}(x,y) = T_i{C_i}(x,y),
\end{equation}
where $T_i$ is the effective transmission time of UAV $i$. Also, ${L_i}(x,y)$ represents the total number of bits transmitted to the user located at $(x,y)$. Note that, the data service offered to each ground user depends on a number of key parameters such as the location of the user and the serving UAV, the bandwidth allocated to the user, and the effective data transmission time of the UAV, $T_i$. Here, we consider the available bandwidth and effective data transmission times as the \textit{resources} used by the UAVs to serviced the ground users. Clearly, the amount of resources that each user can receive depends on several parameters such as the total number of users, cell partitions as well as bandwidth and hover times of the UAVs. Given this model, we next analyze Scenario 1. 

 \section{Scenario 1: Optimal Cell Partitioning for Data Service Maximization with fair resource allocation} \label{Sec}
 
  In this section, \textcolor{black}{our goal is to find the optimal cell partitions that maximize the average data service to the ground users based on the UAVs' hover times and the spatial distribution of the ground users}.
 In this case, each cell partition is assigned to one UAV, and the users within the cell partition must be serviced by the corresponding UAV. We note that, in classical cell partitioning approaches such as Voronoi and weighted Voronoi diagrams \cite{VoronoiDiag}, the spatial distribution of users is not taken into account. As a result, some partitions can be highly congested with users and, hence, each user will receive significantly lower amount of resources than those in less congested partitions. Thus, such classical cell partitioning approaches can lead to a highly unfair data service for the users. In our cell partitioning approach, however, while maximizing the total data service, we ensure that the resources are equally shared among all the users. Hence, our approach avoids creating unbalanced cell partitions
 and, thus, it leads to a higher level of fairness compared to the classical Voronoi approach. In the following, we present the details of our cell partitioning approach based on the UAVs' hover times and the spatial distribution of users.  
 
 Let $\tau_i$ be the hover time of UAV $i$ during which it provides service for the users located in the corresponding cell partition, $\mathcal{A}_i$. The hover time is composed of the effective data transmission time and the control time. 
 To ensure a fair resource allocation policy, we consider the following fairness criterion:\vspace{-0.2cm} 
  \begin{equation} \label{Ti}
  {T _i} = \tau_i - {g_i \left(\int_{{\mathcal{A}_i}} {f(x,y)}\textrm{d}x\textrm{d}y \right)}, \,\,\, \forall i\in \mathcal{M},
 \end{equation}
where $g_i$ is the control time which depends on the number of the users located in $\mathcal{A}_i$. Note that, the average number of users within each cell partition $\mathcal{A}_i$ is linearly proportional to $\int_{{\mathcal{A}_i}} {f(x,y)}\textrm{d}x\textrm{d}y$. In other words, given the spatial distribution of users, $f(x,y)$, and the total number of users, $N$, the average number of users in partition $\mathcal{A}_i$ is equal to $N\int_{{\mathcal{A}_i}} {f(x,y)}\textrm{d}x\textrm{d}y$  \cite{Alonso}.  
From (\ref{Cap}) and (\ref{Loi}) which are used to compute the amount of data service, we can see that the value $T_iB_i$ can be considered as the resources that UAV $i$ uses to service users in $\mathcal{A}_i$. In this case, under a fair resource allocation policy, we should have:
\begin{equation}
\frac{{{T_i}B_i}}{{\int_{{\mathcal{A}_i}} {f(x,y)\textrm{d}x\textrm{d}y} }} = \frac{{{T_j}B_j}}{{\int_{{\mathcal{A}_j}} {f(x,y)\textrm{d}x\textrm{d}y} }},\,\,\, \forall i\ne j \in \mathcal{M}, \label{resource} \footnote{Note that, given hover times of the UAVs, $\tau_i$, $\forall i\in \mathcal{M}$, we can compute $T_i$, $\forall i\in \mathcal{M}$ by solving the system of equations in (\ref{Ti}) and (\ref{resource}). }
\end{equation}
where (\ref{resource}) ensures that a UAV with more resources (bandwidth and hover time) will serve a higher number of users. 


Now, using (\ref{resource}) and considering the fact that $\int_{\mathcal{D}} {f(x,y)}\textrm{d}x\textrm{d}y = \sum\limits_{k = 1}^M {\int_{{\mathcal{A}_k}} {f(x,y)}\textrm{d}x\textrm{d}y }  = 1$, we have the following constraint on the number of users in each partition:
\begin{equation} \label{capacity}
\int_{{\mathcal{A}_i}} {f(x,y)\textrm{d}x\textrm{d}y}  = \frac{{{B_iT_i}}}{{\sum\limits_{k = 1}^M {{B_k T _k}} }},\,\, \forall i\in \mathcal{M}. 
\end{equation} 
As we can see from (\ref{capacity}), the number of users in each generated optimal partition will depend on the UAVs' resources. Clearly,  when the UAVs have the same hover times and bandwidths, (\ref{Ti})-(\ref{capacity}) lead to $\int_{{\mathcal{A}_i}} {f(x,y)\textrm{d}x\textrm{d}y}=\frac{1}{M}$, $\forall i \in \mathcal{M}$. This case implies that the identical UAVs will service equally-loaded cell partitions. 

Given (\ref{Cap}), (\ref{Loi}), and (\ref{capacity}), we can write the average data service at location $(x,y)\in \mathcal{A}_i$ as:
\begin{equation}
{L_i}(x,y) = \frac{{{T_i}B_i}}{{N\int_{{\mathcal{A}_i}} {f(x,y)\textrm{d}x\textrm{d}y} }}{\log _2}\left( {1 + {\gamma _i}(x,y)} \right) = \left( {\frac{1}{N}\sum\limits_{k = 1}^M {{B_k T_k}} } \right){\log _2}\left( {1 + {\gamma _i}(x,y)} \right).
\end{equation}


Now, we formulate an optimization problem for maximizing the average data service by optimal partitioning of the target area. The data service maximization problem is given by:
 \begin{align} \label{Load}
 &\mathop {\max }\limits_{{\mathcal{A}_i},\,  i \in \mathcal{M}} \sum\limits_{i = 1}^M {\int_{{\mathcal{A}_i}} {\left( {\frac{1}{N}\sum\limits_{k = 1}^M {{B_k T_k}} } \right){\log _2}\left( {1 + {\gamma _i}(x,y)} \right)f(x,y)} } \textrm{d}x\textrm{d}y, \\
 \textrm{s.t.}\,\, &\int_{{\mathcal{A}_i}} {f(x,y)} \textrm{d}x\textrm{d}y = \frac{{{B_i T_i}}}{{\sum\limits_{k = 1}^M {{B_k T_k}} }},\,\,\,\forall  i \in \mathcal{M}, \label{f}\\
  & {\gamma _i}(x,y) \ge {\gamma _{th}},\,\,\, {\textrm{if} \,\, (x,y)\in \mathcal{A}_i}, \,\,\, \forall i\in \mathcal{M}, \label{Rcons} \\
  &{\mathcal{A}_l} \cap {\mathcal{A}_m} = \emptyset ,\,\,\,\forall l \ne m \in \mathcal{M},\label{area1}\\
 &\bigcup\limits_{i \in \mathcal{M}} {{\mathcal{A}_i}}  = \mathcal{D},\label{area2}
 \end{align} 
 where (\ref{f}) captures the constraint on the load of each cell partition. Also, (\ref{Rcons}) is the necessary condition for connecting each user to a UAV $i$. (\ref{area1}) and (\ref{area2}) ensure that the cell partitions are disjoint and their union covers the entire target area $\mathcal{D}$.    
 
 
\textcolor{black}{Considering the constraint in (\ref{Rcons}), we define  the function ${q_i}(x,y)= {\left( {\frac{{\gamma _i}(x,y)}{{{\gamma _\textrm{th}}}}} \right)^n }$ with $n$ being a large number (i.e. tends to $+\infty$), and, then, we substract ${q_i}(x,y)$ from the objective function in (\ref{Load}). Clearly, when (\ref{Rcons}) is violated, ${q_i}(x,y)$ tends to $+\infty$ and, hence, point $(x,y)$ will not be assigned to UAV $i$ or equivalently $(x,y) \notin \mathcal{A}_i$.    Therefore, whenever the problem is feasible,} we can remove  (\ref{Rcons}) while penalizing  the objective function in (\ref{Load}) by $q_i(x,y)$. Now, by defining $\lambda=  {\frac{1}{N}\sum\limits_{k = 1}^M {{B_k T_k}} } $, and ${\omega _i} =  \frac{{{B_i T_i}}}{{\sum\limits_{k = 1}^M {{B_k T_k}} }}$, the maximization problem in (\ref{Load}) can be rewritten as the following minimization problem: 
\begin{align} \label{Load2}
&\mathop {\min }\limits_{{\mathcal{A}_i},\,  i \in \mathcal{M}} \sum\limits_{i = 1}^M {\int_{{\mathcal{A}_i}}-{\left(\lambda {\log _2}\left( {1 + {\gamma _i}(x,y)} \right)-q_i(x,y)\right) f(x,y)}} \textrm{d}x\textrm{d}y, \\
\textrm{s.t.}\,\, &\int_{{\mathcal{A}_i}} {f(x,y)} \textrm{d}x\textrm{d}y =\omega_i,\,\,\,\forall  i \in \mathcal{M}, \label{Constr} \\
&{\mathcal{A}_l} \cap {\mathcal{A}_m} = \emptyset ,\,\,\,\forall l \ne m \in \mathcal{M},\\
&\bigcup\limits_{i \in \mathcal{M}} {{\mathcal{A}_i}}  = \mathcal{D}.
\end{align} 

Solving the optimization problem in (\ref{Load2}) is challenging due to various reasons. First, the optimization variables $\mathcal{A}_i$, $ \forall i \in \mathcal{M}$, are sets of continuous partitions (as we have a continuous area) which are mutually dependent. Second, to perfectly capture the spatial distribution of users, $f(x,y)$ is considered to be a generic function of $x$ and $y$ and, this leads to the complexity of the given two-fold integrations. In addition, due to the constraints given in (\ref{Constr}), finding $\mathcal{A}_i$ becomes more challenging. To solve the optimization problem in (\ref{Load2}), next, we model the problem by exploiting \emph{optimal transport theory} \cite{villani}. 

\subsection{Optimal Transport Theory: Preliminaries}
Here, we present some primary results from optimal transport theory which will be used in the next subsection to derive the optimal cell partitions. Optimal transport theory goes back to the Monge's problem in 1781 which is stated as follows \cite{villani}. Given piles of sands and holes with the same volume, what is the best move (transport map) to entirely fill up the holes with the minimum total transportation cost. 
In general, this theory aims to find the optimal matching between two sets of points that minimizes the costs associated with the matching between the sets. These sets can be either discrete or continuous, with arbitrary distributions (weights). Mathematically, the Monge optimal transport problem can be written as follows. Given two probability distributions $f_1$ on $X \subset \mathds{R}^n$, and $f_2$ on $Y \subset \mathds{R}^n$, find the optimal transport map $T$ from $f_1$ to $f_2$  that minimizes the following problem:\vspace{-0.4cm}
\begin{equation}
{\mathop {\min }\limits_T \int_X {c\left( {\boldsymbol{x},T(\boldsymbol{x})} \right)} f_1(\boldsymbol{x})\textrm{d}\boldsymbol{x};\,\,T: X \to Y},
\end{equation}
where $c(\boldsymbol{x},T(\boldsymbol{x}))$ denotes the cost of transporting a unit mass from a location coordinate $\boldsymbol{x}\in X$ to a location $\boldsymbol{y}=T(\boldsymbol{x})\in Y$. Also, as shown in Fig.\,\ref{MAP}, $f_1$ and $f_2$ are the source and destination probability distributions. 

\begin{figure}[!t]
	\begin{center}
		\vspace{-0.1cm}
		\includegraphics[width=6.5cm]{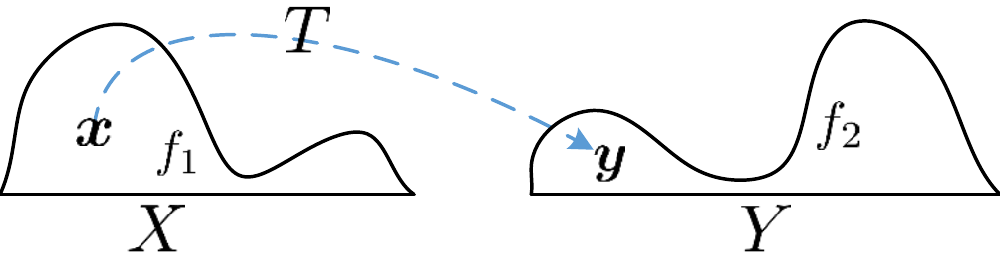}
		\vspace{-0.3cm}
		\caption{  Transport map between two probability distributions.}
		\label{MAP}\vspace{-1.2 cm}
	\end{center}
\end{figure} 
Solving the Monge's problem is challenging due its high non-linear structure \cite{villani}, and the fact that it does not necessarily admit  a solution as each point of the source distribution must be mapped to only one location at the destination. However, Kantorovich relaxed this problem by using transport plans instead of maps, in which one point can go to multiple destination points. The relaxed Monge's problem is called Monge-Kantorovich problem which is written as\cite{villani}:\vspace{-0.15 cm}
\begin{align} \label{Kantor}
&\mathop {\min }\limits_\pi \int_{X \times Y} {c\left( {\boldsymbol{x},\boldsymbol{y}} \right)} \textrm{d}\pi (\boldsymbol{x},\boldsymbol{y}),\\
\textrm{s.t.}\,\,&\int_X {\textrm{d}\pi (\boldsymbol{x},\boldsymbol{y})}  = {f_1}(\boldsymbol{x})\textrm{d}\boldsymbol{x},\,\,\,\int_Y {\textrm{d}\pi (\boldsymbol{x},\boldsymbol{y})}  = {f_2}(\boldsymbol{y})\textrm{d}\boldsymbol{y},
\end{align}
where $\pi$ represents the transport plan which is the probability distribution on $X \times Y$ whose marginals are $f_1$ and $f_2$. 

The Monge-Kantorovich problem has two main advantages compared to the Monge's problem. First, it admits a solution for any semi-continuous cost function. Second, there is a dual formulation for the Monge-Kantorovich problem that can lead to a tractable solution. The duality theorem is stated as \cite{villani} and \cite{Flippo}:

\textbf{Kantrovich Duality Theorem}:
Given the Monge-Kantorovich problem in (\ref{Kantor}) with two probability measures $f_1$ on $X \subset \mathds{R}^n$, and $f_2$ on $Y \subset \mathds{R}^n$, and any lower semi-continuous cost function $c(\boldsymbol{x},\boldsymbol{y})$, the following equality holds:
\begin{align}\label{Duality}
&\mathop {\min }\limits_\pi \int_{X \times Y} {c\left( {\boldsymbol{x},\boldsymbol{y}} \right)} \textrm{d}\pi (\boldsymbol{x},\boldsymbol{y})\\&= \max\limits_{\varphi ,\psi}\left\{ {\int_X {\varphi (\boldsymbol{x}){f_1}(\boldsymbol{x})\textrm{d}\boldsymbol{x}}  + \int_Y {\psi (\boldsymbol{y}){f_2}(\boldsymbol{y})\textrm{d}\boldsymbol{y}}\, ;\,\varphi (\boldsymbol{x}) + \psi (\boldsymbol{y}) \le c(\boldsymbol{x},\boldsymbol{y})} , \, \forall(\boldsymbol{x},\boldsymbol{y})\in X\times Y\right\},
\end{align} 
where $\varphi (\boldsymbol{x})$ and $\psi (\boldsymbol{y})$ are Kantorovich potential functions. As discussed in \cite{villani}, this duality theorem provides a tractable framework for solving the optimal transport problems. In particular, we will use this theorem to tackle our optimization problem in (\ref{Load2}).

We note that, in general, the solutions for the  Monge-Kantorovich problem do not coincide with the Monge's problem. Nevertheless, when the source distribution, $f_1$, and the cost function are continuous, these two problems are equivalent \cite{ambros}. In addition, the optimal transport map, $T: \boldsymbol{x} \to \boldsymbol{y}$, is linked with the optimal Kantorovich potential functions by:\vspace{-0.2cm} 
 \begin{equation}\label{map}
T(\boldsymbol{x}) = \left\{ {\boldsymbol{y}|{\varphi ^*}(\boldsymbol{x}) + {\psi ^*}(\boldsymbol{y}) = c(\boldsymbol{x},\boldsymbol{y})} \right\},
\end{equation}
where ${\varphi ^*}(\boldsymbol{x})$ and ${\psi ^*}(\boldsymbol{y})$ are the optimal potential functions corresponding dual formulation of the Monge-Kantorovich problem.

Given this optimal transport framework, we can solve our optimization problem in (\ref{Load2}). In particular, we model this problem as a semi-discrete optimal transport problem in which the source measure (users' distribution) is continuous while the destination (UAVs' distribution)~is~discrete.\vspace{-0.2cm}  
\subsection{Optimal Cell Partitioning}
   Using optimal transport theory, we can find the optimal cell partitions, $\mathcal{A}_i$, for which the average total data service is maximized. In our model, users have a continuous distribution, and the locations of the UAVs can be considered as discrete points. Then, the optimal cell partitions are obtained by optimally mapping the users to the UAVs. In fact, given (\ref{Load2}), the cell partitions are related to the transport map by \cite{Crippa}:
  \begin{equation} \label{cell}
  \left\{T(\boldsymbol{v}) = \sum\limits_{i \in \mathcal{M}} {{\boldsymbol{s}_i}{\mathds{1}_{{\mathcal{A}_i}}}(\boldsymbol{v})};  \int_{{\mathcal{A}_i}} {f(x,y)\textrm{d}x\textrm{d}y} = {\omega_i} \right\},
  \end{equation}
  where ${\omega _i} =  \frac{{{B_i T_i}}}{{\sum\limits_{k = 1}^M {{B_k T_k}} }}$, as given in (\ref{Constr}), is directly related to the hover time and the bandwidth of the UAVs. Also, $\mathds{1}_{\mathcal{A}_i}(\boldsymbol{v})$ is the indicator function which is equal to 1 if $\boldsymbol{v} \in {\mathcal{A}_i}$, and 0 otherwise.
  
 Therefore, the optimization problem in (\ref{Load2}) can be cast within the optimal transport framework as follows.  Given a continuous probability measure $f$ of users, and a discrete probability measure $\Gamma = \sum\limits_{i \in \mathcal{M}} {{\omega_i}{\delta _{{\boldsymbol{s}_i}}}}$  corresponding to the UAVs, we must find the optimal transport map for which $\int_\mathcal{D} {J\left( {\boldsymbol{v},T(\boldsymbol{v})} \right)} f(x,y)\textrm{d}x\textrm{d}y$ is minimized. In this case, $\delta _{{\boldsymbol{s}_i}}$ is the Dirac function, and $J$ is the transportation cost function which is used in (\ref{Load2}) and is given by:\vspace{-0.2cm} 
\begin{equation}
J(\boldsymbol{v},{\boldsymbol{s}_i}) = J(x,y,{\boldsymbol{s}_i})= {q_i}(x,y) - \lambda {\log _2}\left( {1 + {\gamma _i}(x,y)}\right).
\end{equation}    
 
 Clearly, the cost function, $J$, and the source distribution, $f$, are continuous. As a result, the Monge's problem coincides with the Monge-Kantorovich problem. Next, we propose a solution to (\ref{Load2}) by exploiting the dual formulation of the Monge-Kantorovich problem.

 \begin{theorem} \label{Theor1}
 	 \normalfont
  The optimization problem in (\ref{Load2}) is equivalent to the following unconstrained maximization problem:\vspace{-0.4cm} 
 \begin{equation}\label{F}
\mathop {\max }\limits_{{\psi _i},i \in M} \left\{ {F(\boldsymbol{\psi}^T) = \sum\limits_{i = 1}^M {{\psi _i}{\omega _i}}  + \int_\mathcal{D} {{\psi ^c}(x,y)f(x,y)\textrm{d}x\textrm{d}y} } \right\},
 \end{equation}
where $\boldsymbol{\psi}^T$ is a vector of variables ${{\psi _i}, \forall i \in M}$, and ${\psi ^c}(x,y) = \mathop {\inf }\limits_i J(x,y,{\boldsymbol{s}_i}) - {\psi _i}$.  
 \end{theorem}

\begin{proof}
We use the Kantorovich duality theorem (\ref{Duality}) in which $f(x,y)$ and $\Gamma = \sum\limits_{i \in \mathcal{M}} {{\omega_i}{\delta _{{\boldsymbol{s}_i}}}}$ are two probability measures, and $J(\boldsymbol{v},\boldsymbol{s})$ is the cost function. Clearly, due to the continuity of $f(x,y)$ and $J(\boldsymbol{v},\boldsymbol{s}_i)$, the Monge's problem is equivalent to the Monge-Kantorovich problem.
\begin{align}
&\min \limits_{T}\int_\mathcal{D} {J\left( {\boldsymbol{v},T(\boldsymbol{v})} \right)} f(x,y)\textrm{d}x\textrm{d}y\\&= \max \limits_{\varphi ,\psi}\left\{ {\int_\mathcal{D} {\varphi (\boldsymbol{v}){f}(x,y)\textrm{d}x\textrm{d}y}  + \int_S {\psi (s)\sum\limits_{i \in \mathcal{M}} {{\omega_i}{\delta _{{s-s_i}}}}\textrm{d}s}\, ;\,\varphi (\boldsymbol{v}) + \psi (\boldsymbol{s}) \le J(\boldsymbol{v},\boldsymbol{s})} \right\}\\
&=\max \limits_{\varphi ,\psi}\left\{ \hspace{-0.1cm}{\int_\mathcal{D} {\varphi (x,y){f}(x,y)\textrm{d}x\textrm{d}y}  + {\sum\limits_{i = 1}^M {{\psi( s_i)}{\omega _i}} }\, ;\,\varphi (x,y) + \psi (s_i) \le J(x,y,s_i)},\, \forall i\in\mathcal{M} \right\}\hspace{-0.1cm} \label{dual2}.
\end{align}	
Note that, to maximize (\ref{dual2}) given any $\psi$, we need to choose a maximum value for $\varphi$. Considering the fact that $\varphi (x,y) + \psi (\boldsymbol{s}_i) \le J(x,y,\boldsymbol{s}_i)$ must be satisfied for all $(x,y)\in \mathcal{D}$ and $\boldsymbol{s}_i\in S$, the maximum allowable value of $\varphi$ is given by:\vspace{-0.2 cm} 
\begin{equation}\label{ctrans}
\varphi (x,y) = {\psi ^c}(x,y) = \mathop {\inf }\limits_{{\boldsymbol{s}_i}} J(x,y,{\boldsymbol{s}_i}) - \psi ({\boldsymbol{s}_i}),
\end{equation}	
where ${\psi ^c}$ is called the c-transform of  $\psi$.
Now, considering ${\psi _i} = \psi ({\boldsymbol{s}_i})$, (\ref{dual2}) and (\ref{ctrans}) lead to:
\begin{align} \
\hspace{-0.1cm}&\min\limits_{T}\hspace{-0.1cm}\int_\mathcal{D} {J\left( {\boldsymbol{v},T(\boldsymbol{v})} \right)} f(x,y)\textrm{d}x\textrm{d}y=\hspace{-0.1cm}\mathop {\max }\limits_{{\psi _i},i \in M} \hspace{-0.1cm}\left\{\hspace{-0.1cm} {F(\boldsymbol{\psi}^T) = \sum\limits_{i = 1}^M {{\psi _i}{\omega _i}}  + \int_\mathcal{D} {{\psi ^c}(x,y)f(x,y)\textrm{d}x\textrm{d}y} } \right\},\\
&{\psi ^c}(x,y) = \mathop {\inf }\limits_i J(x,y,{\boldsymbol{s}_i}) - {\psi _i}.\label{psi}
\end{align}	
As a result, the optimization problem in (\ref{Load2}) is reduced to (\ref{F}) with a set of $M$ optimization variables, $\psi_i$, $\forall i\in \mathcal{M}$. This proves the theorem.
\end{proof}
Theorem \ref{Theor1} shows that the complex optimal cell partitioning problem in (\ref{Load2}) can be transformed to a tractable optimization problem with $M$ variables. In other words, by solving (\ref{F}), one can use the optimal values of $\psi_i$, $\forall i\in \mathcal{M}$  to find the optimal cell partitions. Using Theorem \ref{Theor1}, we can further proceed to solve (\ref{Load2}) by presenting the following theorem: 
\begin{theorem} \label{Concave}
	 \normalfont
Given (\ref{F}), $F$ is a concave function of variables $\psi_i$, $i\in \mathcal{M}$. Also, we have:
\begin{equation} \label{Der}
\frac{{\partial F}}{{\partial {\psi _i}}} = {\omega _i} - \int_{{\mathcal{D}_i}} {f(x,y)\textrm{d}x\textrm{d}y},
\end{equation}
where ${\mathcal{D}_i} = \left\{ {(x,y)|J(x,y,{\boldsymbol{s}_i}) - {\psi _i} \le J(x,y,{\boldsymbol{s}_j}) - {\psi _j},\forall j \ne i} \right\}$.
\end{theorem}

\begin{proof}
Clearly, $\sum\limits_{i = 1}^M {{\psi _i}{\omega _i}}$ is a linear function of $\psi_i$. Also, given any $i\in \mathcal{M}$, $J(x,y,{\boldsymbol{s}_i}) - {\psi _i}$ is a linear function of $\psi_i$. Let $z(\boldsymbol{\psi}^T)= \mathop {\inf }\limits_i J(x,y,{\boldsymbol{s}_i}) - {\psi _i}$ with $\boldsymbol{\psi}^T$ being a vector of all variables $\psi_i$, $i\in \mathcal{M}$. Then, we can observe that the hypograph of $z(\boldsymbol{\psi}^T)$, a set of points below $z(\boldsymbol{\psi}^T)$, is a convex set. Subsequently, considering the fact that a function is concave if and only if its hypograph is convex, we prove the concavity of $z(\boldsymbol{\psi}^T)$. Finally, since multiplying $z(\boldsymbol{\psi}^T)$ by a positive probability density function $f(x,y)$, and taking integration over $(x,y)$ does not violate the concavity, $F$ is also a concave function of $\boldsymbol{\psi}^T$. 

To find the derivative of $F$ with respect to $\psi_i$, we first compute $\frac{{\partial {\psi ^c}}}{{\partial {\psi _i}}}$. Clearly, based on (\ref{psi}), we have:\vspace{-0.3cm} 
\begin{align}
\frac{{\partial {\psi ^c}}}{{\partial {\psi _i}}} = \left\{ \begin{array}{l}
\hspace{-0.1cm}- 1,\hspace{1cm} \textrm{if}\,\,\,\, J(x,y,{\boldsymbol{s}_i}) - {\psi _i} \le J(x,y,{\boldsymbol{s}_j}) - {\psi _j},\,\,\forall j \ne i,\\
0,\hspace{1cm} \textrm{otherwise}.
\end{array} \right.
\end{align}
Then, by defining ${\mathcal{D}_i} = \left\{ {(x,y)|J(x,y,{\boldsymbol{s}_i}) - {\psi _i} \le J(x,y,{\boldsymbol{s}_j}) - {\psi _j},\forall j \ne i} \right\}$, the derivative of $F$, given in (\ref{F}), will be:
\begin{equation}
\frac{{\partial F}}{{\partial {\psi _i}}} = {\omega _i} - \int_{{\mathcal{D}_i}} {f(x,y)\textrm{d}x\textrm{d}y}\,.
\end{equation}
This proves the theorem.
\end{proof}

\setlength\textfloatsep{0.6\baselineskip plus 3pt minus 2pt}
\begin{algorithm} [t]
	\begin{small}
		\caption{Gradient method for optimal cell partitioning }\label{Gradient}
		\begin{algorithmic}[1] 
			\State {\textbf{Inputs:} $f(x,y)$, $\rho$,  $\tau_i$,  $\boldsymbol{s}_i$, $\forall i \in \mathcal{M}$.}
			\State \textbf{Outputs:} $\psi_i^*$, $\mathcal{A}_i$, $\forall i \in \mathcal{M}$. 
			\State {Set initial values for $\boldsymbol{\psi}^T_t$, ($t=1$) }.
			\While {${\left\| {\nabla F(\boldsymbol{\psi}^T_t)} \right\|_2} > \rho $} \label{stopping}
			\State {Set $k=1$, {$\epsilon_1=1.$}}
			\State {Update  $\boldsymbol{\psi}^T_{t+1}=\boldsymbol{\psi}^T_{t}+\varepsilon_k \nabla F(\boldsymbol{\psi}^T_{t})$, $k,t \in \mathds{N}$}. \label{Update} 
			\If {$F(\boldsymbol{\psi}^T_{t})< F(\boldsymbol{\psi}^T_{t+1})$} \label{stepsize1}
			\State {Go to Step \ref{if1} }.
			\Else 
			\State {Go to Step \ref{if2} }.
			\EndIf
			\While {$F(\boldsymbol{\psi}^T_{t})< F(\boldsymbol{\psi}^T_{t+1})$} \label{if1}
			\State {$k\to k+1$}.
			\State {$\epsilon_{k}=2^{k-1}\epsilon_{1}$}.
			\State {Update  $\boldsymbol{\psi}^T_{t+1}=\boldsymbol{\psi}^T_{t}+\varepsilon_k \nabla F(\boldsymbol{\psi}^T_{t})$, $k,t \in \mathds{N}$}.	
			\EndWhile	
			\While {$F(\boldsymbol{\psi}^T_{t})> F(\boldsymbol{\psi}^T_{t+1})$} \label{if2}
			\State {$k\to k+1$}.
			\State {$\epsilon_{k}=2^{-k+1}\epsilon_{1}$}.
			\State {Update  $\boldsymbol{\psi}^T_{t+1}=\boldsymbol{\psi}^T_{t}+\varepsilon_k \nabla F(\boldsymbol{\psi}^T_{t})$, $k,t \in \mathds{N}$}.			
			\EndWhile \label{stepsize2}			
			\State {$t\to t+1$}.
			\EndWhile	
			\State 	{$\psi _i^*=\boldsymbol{\psi}^T_{t}(i)$, $\forall i \in \mathcal{M}$.} \label{op}
			\State {${\mathcal{A}_i} = \left\{ {(x,y)|J(x,y,{\boldsymbol{s}_i}) - {\psi _i^*} \le J(x,y,{\boldsymbol{s}_j}) - {\psi _j^*},\forall j \ne i} \right\}$, $\forall i \in \mathcal{M}$.}\label{A}
		\end{algorithmic} 
	\end{small} 
\end{algorithm}

Theorem \ref{Concave} shows the concavity of $F$ as a function of  $\boldsymbol{\psi}^T$. Thus, the optimal values for variables $\psi_i$, $\forall i\in \mathcal{M}$, can be obtained by maximizing $F$. Then, given the optimal $\psi_i$, $\forall i\in \mathcal{M}$, equations (\ref{map}), and (\ref{cell}) are used to determine the optimal cell partitions corresponding to the optimization problem in (\ref{Load2}). 
 In this case, using the first derivative of $F$ provided in (\ref{Der}), we propose a gradient-based method to determine the optimal vector $\boldsymbol{\psi}^T$ that leads to the optimal cell partitions. Here, using the gradient descent method is simple in terms of implementation and does not require computing the Hessian matrix of $F$ which is needed in the Newton methods. In fact, given the intractable expression of $F$ in (\ref{F}), finding its second derivative is challenging and, thus, we adopt the gradient-based approach.

The proposed algorithm for finding the optimal cell partitions is shown as Algorithm 1 and proceeds as follows. The inputs are the distribution of users, hover times, locations of the UAVs, and $\rho > 0$ which is the threshold based on which the algorithm stops. In Algorithm \ref{Gradient}, we first initialize vector $\boldsymbol{\psi}^T_{t}$ with $t$ being the iteration number. Next, using (\ref{Der}), we compute $\nabla F(\boldsymbol{\psi}^T_t)$.  In Step \ref{Update}, we update $\boldsymbol{\psi}^T_t$ using step size $\epsilon_k$. The appropriate step size at each iteration is determined through Steps \ref{stepsize1} to \ref{stepsize2}. In this case, the algorithm stops when the condition in $\ref{stopping}$ is not satisfied. Clearly, due to the concavity of $F$, the optimal solution to (\ref{F}) is attained. Finally, based on the  optimal vector $\boldsymbol{\psi}^T$, the optimal cell partitions are determined using Steps \ref{op} and~\ref{A}.   

In summary, we proposed a framework for maximizing the average data service to ground users while considering some level fairness between the users. To this end, we used tools from optimal transport theory to determine the optimal cell partition associated to each UAV that services the users in the cell partition. 
In the next section, we investigate Scenario 2 in which the minimum average hover times of UAVs needed for completely servicing the users are determined.   

\section{Scenario 2: Minimum Hover Time For meeting Load Requirements} 
Our goal is to meet the load requirements of the ground users while minimizing the average hover times of the UAVs. In particular, given the demand of each user, we find the minimum required average hover time of the UAVs to completely serve the users. The hover time of each UAV depends on the distribution and load of the users, bandwidth allocation between users, and the cell partition that is assigned to the UAV. Next, we first derive an expression for the average hover time needed to serve any arbitrary cell partition under an optimal bandwidth allocation between the users. Then, we determine the optimal cell partitions for which the average hover times required for completely servicing the entire target area is minimized. We note that, given a specified partition $\mathcal{A}_i$, the hover time of the UAV $i$ needed for serving the users depends on the bandwidth allocation strategy. Hence, we first derive the minimum average hover time of each UAV that can be attained by optimal bandwidth allocation to the users in the given cell~partition. 
\begin{proposition}
	 \normalfont
Let $u(x,y)$ be the load (in bits) of a user located at $(x,y)$. The minimum average hover time of UAV $i$ for serving partition $\mathcal{A}_i$ that can be achieved by optimally allocating the bandwidth to the users, is given by:
\begin{equation} \label{tau}
{\tau _i} = \int_{{\mathcal{A}_i}} {\frac{{N u(x,y)}}{C_i^{B_i} (x,y)}f(x,y)} \textrm{d}x\textrm{d}y + {g_i}\left( {\int_{{\mathcal{A}_i}} {f(x,y)} \textrm{d}x\textrm{d}y} \right),
\end{equation}
where $N$ is the total number of users, $C_i^{B_i}={B_i{\log _2}\left(1+\gamma_i(x,y)\right)}$, and $g_i ( .)$ is the additional control time which is a function of the number of users in the cell partition.
\end{proposition}\vspace{-0.2 cm} 
\begin{proof}
 Let $\mathcal{Y}$ be an arbitrary set of $Y$ users in a cell partition $\mathcal{A}_i$. We denote the load and bandwidth allocated to user $r$ by $u_r$ and $W_r$. Then, the time needed to serve user $r$ is given by:\vspace{-0.2 cm} 
 \begin{equation}
 {t_r} = \frac{{{u_r}}}{{{W_r}{E_r}}},
 \end{equation}
 where $E_r$ is the spectral efficiency (bit/s/Hz) at the user's location. Clearly, to completely serve all the users in $\mathcal{A}_i$, the hover time of UAV $i$ must be $\max{t_r}+g_\mathcal{Y}$, with $g_\mathcal{Y}$ being the additional control time. In this case, the hover time can be minimized by an optimal bandwidth allocation as follows:\vspace{-0.7cm}
 \begin{align} \label{minmax}
 &\mathop {\min \max }\limits_{{W_r},\,r = 1,...,Y} {t_r} + {g_\mathcal{Y}},\\
 \textrm{s.t.}\,\,\, &\sum\limits_{r = 1}^Y {{W_r}}  = B_i.
 \end{align}
 The minmax problem in (\ref{minmax}) can be transformed to:\vspace{-0.3 cm} 
 \begin{align}
 &\min Z + {g_\mathcal{Y}}, \\
 \textrm{s.t.}\,\,\, & Z \ge \frac{{{u_r}}}{{{W_r}{E_r}}},\,\, \forall r\in \mathcal{Y}\label{C1},\\
 &\sum\limits_{r = 1}^Y {{W_r}}  = B_i \label{C2}.
 \end{align}
 Now, using (\ref{C1}) and (\ref{C2}), we have $Z \ge \frac{1}{B_i}\sum\limits_{r = 1}^Y {\frac{{{u_r}}}{{{E_r}}}}$. Hence, the minimum hover time under an optimal bandwidth allocation will be $\sum\limits_{r = 1}^Y {\frac{{{u_r}}}{{B_i{E_r}}}}  + {g_\mathcal{Y}}$. Furthermore, it can be shown that ${W_r} = \frac{{B_i{u_r}}}{{{E_r}}}/\sum\limits_{k = 1}^Y {\frac{{{u_k}}}{{{E_k}}}}$ is an optimal bandwidth allocation to user $r$. Clearly, $\sum\limits_{r = 1}^Y {\frac{{{u_r}}}{{B_i{E_r}}}}$ is also equal to the total time needed for sequentially serving the users using the entire bandwidth $B_i$.  Therefore, given a cell partition $\mathcal{A}_i$ and the users' distribution, $f(x,y)$, the minimum average hover time of UAV $i$ can be given by: 
 \begin{equation} \label{tau2}
 {\tau _i} = \int_{{\mathcal{A}_i}} {\frac{{Nu(x,y)}}{{{B_i{\log _2}\left(1+\gamma_i(x,y)\right)}}}f(x,y)} \textrm{d}x\textrm{d}y + {g_i}\left( {\int_{{\mathcal{A}_i}} {f(x,y)} \textrm{d}x\textrm{d}y} \right).
 \end{equation}	
Finally, considering $C_i^{B_i}={B_i{\log _2}\left(1+\gamma_i(x,y)\right)}$, this proposition is proved.	
\end{proof}\vspace{0.1cm}

From (\ref{tau}), we can see that the hover time increases as the number of users increases. In fact, for a higher number users, both the total data transmission time and the additional control time increase. Moreover, the hover time increases as the load of the users increases. In particular, for an equal load of users, the hover time increases sublinearly \textcolor{black}{by increasing the load as the control time does not depend on the load here}. From (\ref{tau}), we can see that the hover time can be reduced by increasing the transmission rate. In addition, (\ref{tau}) implies that the UAV must hover for a longer time over a partition with higher users' density (i.e. $f(x,y)$).


Next, we minimize the total average hover time by solving the following optimization problem:
 \begin{align} \label{hover}
 &\mathop {\min }\limits_{{\mathcal{A}_i},\,  i \in \mathcal{M}} \sum\limits_{i = 1}^M {\int_{{\mathcal{A}_i}} {{\frac{{Nu(x,y)}}{{{C_i^{B_i}}(x,y)}}}f(x,y)} } \textrm{d}x\textrm{d}y + {g_i}\left( {\int_{{\mathcal{A}_i}} {f(x,y)} \textrm{d}x\textrm{d}y} \right),\\
 \textrm{s.t.}\,\, \,\,
 & {\gamma _i}(x,y) \ge {\gamma _\textrm{th}},\,\,\, {\textrm{if} \,\, (x,y)\in \mathcal{A}_i}, \,\,\, \forall i\in \mathcal{M}, \label{D2} \\
 &{\mathcal{A}_l} \cap {\mathcal{A}_m} = \emptyset ,\,\,\,\forall l \ne m \in \mathcal{M},\label{D3}\\
 &\bigcup\limits_{i \in \mathcal{M}} {{\mathcal{A}_i}}  = \mathcal{D},\label{D4}
 \end{align} 
where the objective function in (\ref{hover}) represents the total average hover time needed for providing service for the target area. Following an approach that is similar to the one we used in (\ref{Load}), our optimization problem in (\ref{hover}) can be rewritten as:
 \begin{align} \label{hover2}
 &\mathop {\min}\limits_{{\mathcal{A}_i},\,  i \in \mathcal{M}} \sum\limits_{i = 1}^M {\int_{{\mathcal{A}_i}} {{\left[\frac{Nu(x,y)}{{{C_i^{B_i}}(x,y)}}+q_i(x,y)\right]}f(x,y)} } \textrm{d}x\textrm{d}y + {g_i}\left( {\int_{{\mathcal{A}_i}} {f(x,y)} \textrm{d}x\textrm{d}y} \right),\\
 \textrm{s.t.}\,\, \,\, 
 &{\mathcal{A}_l} \cap {\mathcal{A}_m} = \emptyset ,\,\,\,\forall l \ne m \in \mathcal{M},\label{E3}\\
 &\bigcup\limits_{i \in \mathcal{M}} {{\mathcal{A}_i}}  = \mathcal{D},\label{E4}
 \end{align} 
where ${q_i}(x,y)= {\left( {\frac{{\gamma _i}(x,y)}{{{\gamma _\textrm{th}}}}} \right)^n }$ with $n \to +\infty$.

Solving (\ref{hover2}) is challenging as  the optimization variables $\mathcal{A}_i$, $ \forall i \in \mathcal{M}$, are sets of continuous partitions which are mutually dependent. Moreover, since $g_i$ is a generic function of $\mathcal{A}_i$ and  $f(x,y)$, this problem is intractable. Next, we use optimal transport theory to completely characterize the optimal solution. To this end, we first prove the existence of the solution to (\ref{hover2}) for any semi-continuous function $g_i$, $\forall i\in \mathcal{M}$. Note that, in general, (\ref{hover2}) does not necessary admit an optimal solution when the semi-continuity of $g_i$ does not hold.   

\begin{proposition}\label{exist}
	\normalfont
	The optimization problem in (\ref{hover2}) generally admits an optimal solution.
\end {proposition}\vspace{-0.2 cm} 
	\begin{proof}
		Let ${a_i} = \int_{{\mathcal{A}_i}} {f(x,y)\textrm{d}x\textrm{d}y}$, then we also define a unit simplex as follows:
		\begin{equation}
		E = \left\{ {\boldsymbol{a} = \left( {{a_1},{a_2},...,{a_{M}}} \right) \in {\mathds{R}^{M}};\,\,\sum\limits_{k = 1}^{M} {{a_i}= 1}, {a_i} \ge 0, \forall i\in \mathcal{M} } \right\}.
		\end{equation}
		
		Clearly, given any vector $\boldsymbol{a}$, problem (\ref{hover2}) can be considered as a classical semi-discrete optimal transport problem. In particular, considering  $f(\boldsymbol{v})=f(x,y)$, and $c\left( {\boldsymbol{v},{\boldsymbol{s}_i}} \right) = \frac{u(x,y)}{{{C_i^{B_i}}(x,y)}}+q_i(x,y)$, (\ref{hover2}) can be transformed to:\vspace{-0.35 cm} 
		\begin{equation} \label{Opt5}
		\mathop {\min }\limits_T \int_\mathcal{D} {c\left( {\boldsymbol{v},\boldsymbol{s}} \right)} f(\boldsymbol{v})\textrm{d}\boldsymbol{v}, \,\, \boldsymbol{s}=T(\boldsymbol{v}),
		\end{equation}
		where $T$ is the transport map which is associated to cell partitions $\mathcal{A}_i$ by :\vspace{-0.15cm}
		\begin{equation}
		\left\{T(\boldsymbol{v}) = \sum\limits_{i=1}^ {M} {\boldsymbol{s}_i{\mathds{1}_{{\mathcal{A}_i}}}(\boldsymbol{v})}; \int_{{\mathcal{A}_i}} f(\boldsymbol{v})\textrm{d}\boldsymbol{v}={a_i}\right\}.
		\end{equation}

	Clearly, as discussed in Section \ref{Sec}, the optimal transport problem in (\ref{Opt5}) admits a solution. Hence, for any $\boldsymbol{a}\in E$, the problem in (\ref{hover2}) has an optimal solution. Since $E$ is a unit simplex in $\mathds{R}^{M}$ which is a non-empty and compact set, the problem admits an optimal solution over the entire $E$.
	Thus, the proposition is proved.
\end{proof}

Next, we completely characterize the solution space of problem (\ref{hover2}) which allows finding the optimal cell partitioning and the average hover time of each UAV.\vspace{-0.3cm}
\begin{theorem} \label{Optimal}
	\normalfont
	The optimal hover time of UAV $i$ required to completely service the target area is given by:	\vspace{-0.2cm} 
	\begin{equation}
	{\tau_i^*} = \int_{{\mathcal{A}_i^*}} {\frac{{Nu(x,y)}}{{{C_i^{B_i}}(x,y)}}f(x,y)} \textrm{d}x\textrm{d}y + {g_i}\left( {\int_{{\mathcal{A}_i^*}} {f(x,y)} \textrm{d}x\textrm{d}y} \right),
	\end{equation}
	where $\mathcal{A}_i^*$ is the optimal cell partition given by:	
\begin{equation}
		\hspace{-0.00cm} {\mathcal{A}_i^*} \hspace{-0.06cm}=	\hspace{-0.1cm} \left\{\hspace{-0.05cm}\hspace{-0.07cm}{(x,y)|\, \frac{Nu(x,y)}{C_i^{B_i}(x,y)}+q_i(x,y)+g'_i(a_i) \hspace{-0.05cm} \le \hspace{-0.03cm} \frac{Nu(x,y)}{C_j^{B_j}(x,y)}+q_j(x,y)+g'_j(a_j),\forall j \ne i \in \mathcal{M}	\hspace{-0.07cm}} \right\}	\hspace{-0.1cm}, \hspace{-0.17cm} \label{cells}
	\end{equation}
	where ${a_i} = \int_{{\mathcal{A}_i}} {f(x,y)} \textrm{d}x\textrm{d}y$, and $N$ is the total number of users.
\end{theorem}
\begin {proof}
See Appendix A.
\end{proof}
Using Theorem \ref{Optimal}, we can find the optimal cell partitions as well as the minimum hover time needed to completely service the users. In fact, the target area is optimally partitioned in a way that the average hover time that the UAVs use to serve their users is minimized. Note that, for the special case where $g_i'=0$, the result in (\ref{cellPar}) corresponds to the classical weighted Voronoi diagram. In this case, the users are assigned to the UAVs based on the maximum received signal strength criterion. Consequently, the users can be served with a maximum rate and, hence, the total required hover time is minimized. However, in general, the classical weighted Voronoi diagram is not optimal \cite{Alonso} as the effect of control time is ignored while generating cell partitions. 
 
From (\ref{cells}), we can see that there is a mutual dependency between $a_i$ and $\mathcal{A}_i$, $\forall i \in \mathcal{M}$. \textcolor{black}{Therefore, solving (\ref{cells}) does not have an explicit form and, hence, an iterative-based approach is needed to find a solution to (\ref{cells})}. Next, given the results of Theorem \ref{Optimal}, we present an iterative algorithm based on \cite{Crippa} and shown in Algorithm 2, that solves (\ref{cells}) and finds the optimal cell partitions and the average hover time of each UAV. This algorithm guarantees the convergence to the optimal solution within a reasonable number of iterations \cite{Crippa}. \textcolor{black}{In addition, this algorithm is practical to implement as its complexity grows linearly with the size of the area $\mathcal{D}$}.   

\begin{algorithm}[!t]
	\begin{small}
		\caption{Iterative algorithm for optimal cell partitions and hover times}
		\label{Fixed}
		\begin{algorithmic}[1] 
			\State \textbf{Inputs:} $f(x,y)$, $u(x,y)$, $Z$,  $g_i$,  $\boldsymbol{s}_i$, $\forall i \in \mathcal{M}$.
			\State \textbf{Outputs:} $\mathcal{A}_i^*$, $\tau_i^*$, $\forall i \in \mathcal{M}$. 
			\State {Set $t=1$, generate an initial cell partitions $\mathcal{A}_i^{(t)}$, and set $\phi _i^{(t)}(x,y)=0$, $\forall i \in\mathcal{M}$}.
			
			\While {$t<Z$} 
			
			\State {Compute $\phi _i^{(t + 1)}(x,y) = \left\{ \begin{array}{l}
				\left( {1 - 1/t} \right)\phi _i^{(t)}(x,y),\hspace{1.8cm} \textrm{if}\,\,(x,y) \in {\mathcal{A}_i}^{(t)},\\
				1 - \left( {1 - 1/t} \right)\left( {1 - \phi _i^{(t)}(x,y)}\right),\,\, \textrm{otherwise}.
				\end{array} \right.$  } \label{phi}

			\State {Compute ${a_i} = \int_\mathcal{D} {\left( {1 - \phi _i^{(t + 1)}(x,y)} \right)f(x,y)} \textrm{d}x\textrm{d}y$, $\forall i \in \mathcal{M}$}. \label{1}

			\State {$t \to t+1$}.\label{2}
			
			\State {Update cell partitions using (\ref{cellPar})}.\label{3}
			

			
			\EndWhile
			\State 	{$\mathcal{A}_i^*= \mathcal{A}_i^{(t)}$,} \label{A*}
			
			\State {Compute $\tau_i^{*}$ using (\ref{tau}) based on $\mathcal{A}_i^*$, $\forall i \in \mathcal{M}$.}\label{tau*}
			
		\end{algorithmic}
	\end{small} 
\end{algorithm}

Algorithm 2 for finding the optimal cell partitions as well as the average hover times proceeds as follows. The inputs are load and distribution of the users, locations of the UAVs, control time function, and the number of iterations, $L$. Here, we use $t$ to represent the iteration number. First, we generate initial cell partitions $\mathcal{A}_i^{(t)}$, and set $\phi _i^{(t)}(x,y)=0$, $\forall i \in \mathcal{M}$, with $\phi _i^{(t)}(x,y)$ being a pre-defined parameter that will be used to update the cell partitions. Next, we update  $\phi _i^{(t+1)}(x,y)$, and compute $a_i$ in step \ref{1}. Then, in step \ref{3}, we update cell partitions by using (\ref{cellPar}). Finally, at the end of the iteration, the optimal cell partitions and the minimum average hover time of the UAVs are determined.\vspace{-0.2cm}

\section{Simulation Results and Analysis}

For our simulations, we consider a rectangular area of size $1000\,\text{m}\times 1000 \,\text{m}$ in which the ground users are distributed according to a two-dimensional truncated Gaussian distribution which is suitable to model a hotspot area and is given by \cite{ghazzai}:\vspace{-0.03cm} 
\begin{equation} %
f(x,y) = \frac{1}{\eta}{\rm{exp}}{\Big( {\frac{{{L_x} - {\mu _x}}}{{\sqrt {2{\sigma _x}} }}} \Big)^2}{\rm{exp}}{\Big( {\frac{{{L_y} - {\mu _y}}}{{\sqrt {2{\sigma _y}} }}} \Big)^2}, \label{Trunc}\vspace{-0.1cm}
\end{equation}
where $\eta= 2\pi {\sigma _x}{\sigma _y}{\rm{erf}}\left( {\frac{{{L_x} - {\mu _x}}}{{\sqrt {2{\sigma _x}} }}} \right){\rm{erf}}\Big( {\frac{{{L_y} - {\mu _y}}}{{\sqrt {2{\sigma _y}} }}} \Big)$, and the size of the area is $L_x\times L_y$. Also, ${\mu _x}$, ${\sigma _x}$, ${\mu _y}$, and ${\sigma _y}$ are the mean and standard deviation values of the $x$ and $y$ coordinates, and ${\rm{erf(}}z{\rm{)}} = \frac{2}{{\sqrt \pi  }}\int\limits_0^z {{e^{ - {t^2}}}{\rm{d}}t}$.
In this case, (${\mu _x}$, ${\mu _y}$) represents the center of the hotspot, and the density of the users around the center is inversely proportional to the values ${\sigma _x}$ and ${\sigma _y}$. In our simulations, we consider $\sigma_x=\sigma_y=\sigma_o$. Note that, although we consider the truncated Gaussian distribution of users, our analysis can also accommodate any any other arbitrary distribution. Moreover, we deploy the UAVs based on a grid-based deployment with an altitude of 200\,m. Unless stated otherwise, we consider a full interference scenario with an interference factor $\beta=1$. For the control time function, we consider ${g_i}(N{a_i}) = \alpha {\left( {N{a_i}} \right)^2}$, with $\alpha$ being an arbitrary constant factor. This function is a  reasonable choice in our model as it is a superlinear function of the number of users and its value can be adjusted by factor $\alpha$. However, any arbitrary continuous control function can also be considered in our model. The simulation parameters are listed in Table I. We compare our results, obtained based on the proposed optimal cell partitioning approach, with the classical weighted Voronoi diagram baseline. Note that, all statistical results are averaged over
a large number of independent runs. Next, we present the results corresponding to Scenario 1 and Scenario 2, separately. \vspace{-0.3cm}

\begin{table}[!t]
	\normalsize
	\begin{center}
		\caption{\small Simulation parameters.}
		\vspace{-0.1cm}
		\label{TableP}
		\resizebox{8.1cm}{!}{
			\begin{tabular}{|c|c|c|}
				\hline
				\textbf{Parameter} & \textbf{Description} & \textbf{Value} \\ \hline \hline
				$f_c$	&     Carrier frequency     &      2\,GHz     \\ \hline 
				$P_i$	&     UAV transmit power     &     0.5\,\textrm{W}    \\ \hline
				
				$N_o$	&     Noise power spectral     &        -170\,dBm/Hz  \\ \hline
				$N$	&     Number of ground users     &   300 \\ \hline
				
				$\mu_\textrm{LoS}$	&     Additional path loss to free space for LoS     &   3\,dB \\ \hline
				
				$\mu_\textrm{NLoS}$	&     Additional path loss to free space for NLoS      &   23\,dB \\ \hline
				
				$B$	&    Bandwidth      &   1\,MHz \\ \hline
				
				$\alpha$	&    Control time factor      &   0.01 \\ \hline
				
				$h$	&     UAV's altitude     &   200\,m \\ \hline
				
				$u$	&     Load per user     &   100\,Mb \\ \hline
				
				$\mu_x, \mu_y$& Mean of the truncated Gaussian distribution & 250\,m, 330\,m\\ \hline
				
				$b_1, b_2$	&     Environmental parameters (dense urban)     &   0.36, 0.21 \cite{HouraniModeling} \\ \hline
			\end{tabular}}
			
		\end{center}\vspace{-0.5cm}
	\end{table}

%

\subsection{Results for Scenario 1}

\begin{figure}[t]
	\centering
\hspace{-0.9cm}	\begin{subfigure}[t]{0.4\textwidth}
		\begin{center}
			\vspace{-0.3cm}
			\includegraphics[width=8.1cm]{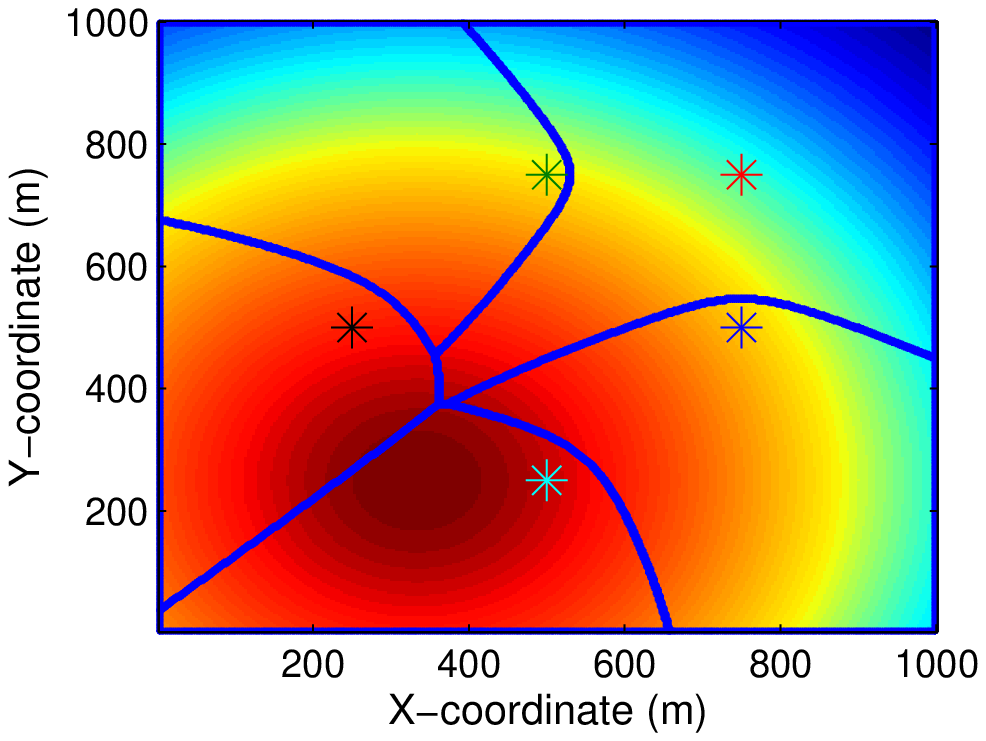}
			\vspace{-0.9cm}
			\caption{ \small Proposed optimal cell partitions.}\vspace{-.2cm}
			\label{Ayy}
		\end{center}
	\end{subfigure}%
	~	\hspace{0.5cm}
	\begin{subfigure}[t]{0.4\textwidth}
		\begin{center}
			\vspace{-0.3cm}
			\includegraphics[width=8.1cm]{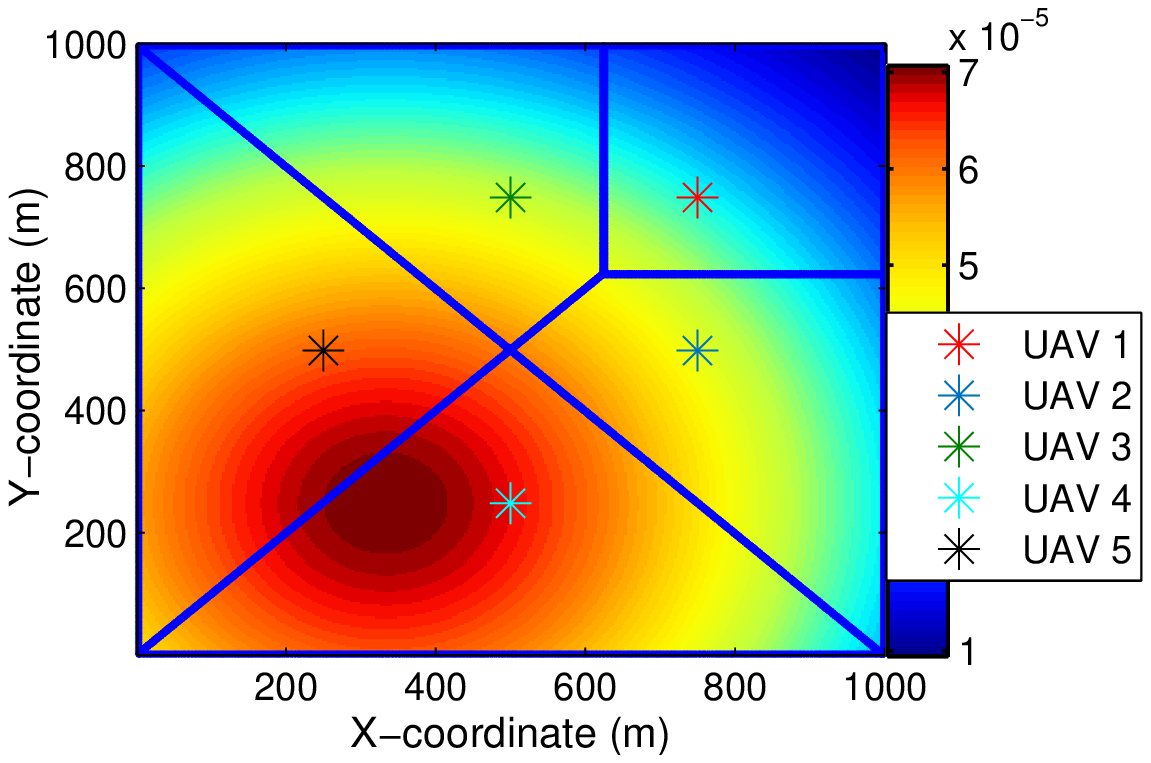}
			\vspace{-0.9cm}
			\caption{ \small Weighted Voronoi diagram.}\vspace{-.2cm}
			\label{Voro}
		\end{center}
	\end{subfigure}\vspace{-0.10cm}
	
	\caption{  Cell partitions associated to UAVs given the non-uniform spatial distribution~of~users. \vspace{-.18cm} 	\label{CellPart}}\vspace{.3cm} 
\end{figure}



Fig.\,\ref{CellPart} shows the proposed optimal cell partitions and the classical weighted Voronoi diagram. In this case, we consider 5 UAVs that provide service for the non-uniformly distributed ground users (truncated Gaussian distribution with $\sigma_o=1000$\,m). Moreover, in Scenario 1, we assume that the maximum hover time of each UAV is 30 minutes which corresponds to the typical hover time for quadcopter UAVs \cite{Quadcopter}.  In Fig.\,\ref{CellPart}, areas shown by a darker color have a higher population density. As we can see from Fig.\,\ref{Voro}, the cell partitions associated with UAVs 4 and 5 have significantly more users than cell partition 1. Therefore, given the limited hover times, users located at cell partitions 4 and 5 cannot be fairly served by UAVs. However, in the proposed optimal cell partitioning case (obtained by Algorithm 1), the cell partitions change such that the average data service under a fair resource allocation constraint is maximized. For instance, as shown in Fig.\,\ref{Ayy}, the size of cell partitions 4 and 5 decreases compared to the weighted Voronoi diagram. As a result, the proposed cell partitions lead to a higher level of fairness among the users than the weighted Voronoi case.

\begin{figure}[!t]
	\begin{center}
		\vspace{-0.1cm}
		\includegraphics[width=8.5cm]{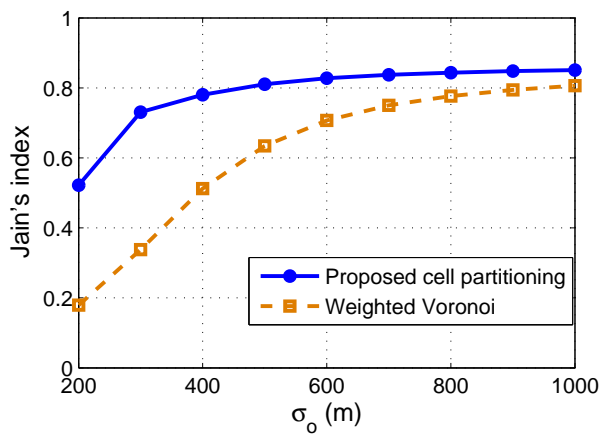}
		\vspace{-0.3cm}
		\caption{ Jain's fairness index for average data service to users.}\vspace{-.32cm}
		\label{Jains}
	\end{center}\vspace{-.12cm}
\end{figure}

To show how fairly the users can be served in different cases of cell partitioning, we use the Jain's fairness index. This metric can be applied to any performance metric such as rate or service load and it is maximized when all users receive an equal service. Here, we compute the Jain's index based on the data service that is offered to each user. The Jain's index is given~by~\cite{jain}:\vspace{-0.1cm}
\begin{equation}
F_\text{Jain}({l_1},{l_2},...,{l_N}) = {{{\left( {\sum\limits_{i = 1}^N {{l_i}} } \right)}^2}}\times\Big({N\sum\limits_{i = 1}^N {{l_i}^2} }\Big)^{-1},
\end{equation}
where $N$ is the number of users, and $l_i$ is the data service to user $i$. Clearly, $1/N \le F_\text{Jain} \le 1$, with $F_\text{Jain}=1/N$ and $F_\text{Jain}=1$ indicating the lowest and  highest level of fairness. 

Fig.\,\ref{Jains} shows the Jain's fairness index for different values of $\sigma_o$ which is given in (\ref{Trunc}). In this figure, as $\sigma_o$ increases the spatial distribution of users becomes closer to a uniform distribution. As we can see from this figure, the minimum Jain's index corresponding to the proposed cell partitioning method is above 0.5. However, in the weighted Voronoi case, it can decrease to 0.18 for a highly non-uniform distribution of users with $\sigma_o=200$\,m. This is due to the fact that, in the Voronoi case, users located in highly congested partitions receive lower service than the partitions with low number of users. In the proposed approach, however, the resources (hover time and bandwidth) are fairly shared between the users thus leading to a higher fairness index. From Fig.\,\ref{Jains} we can also observe that, for higher values of $\sigma_o$ (more uniform distribution), the fairness index for the proposed approach becomes closer to the weighted Voronoi case.


Fig.\,\ref{NumbUsers} shows the average number of users in each cell partition. Clearly, in the Voronoi case, the average number of users per cell significantly varies for different cell partitions. For instance, the average number of users in cell 5 is three times higher than cell 3. Consequently, compared to cell 5, users in cell 3 will receive lower data service from their associated UAV. However, in the proposed approach, the cell partitions associated with the UAVs are formed such that the number of users per cell be proportional to the bandwidth and hover time of the UAVs. In this case, given equal bandwidth and hover times of UAVs, the cell partitions contains an equal number of users. Therefore, our approach avoids generating unbalanced cell partitions and, hence, it leads to a higher level of fairness compared to the classical Voronoi approach.    

\begin{figure}[!t]
	\begin{center}
		\vspace{-0.1cm}
		\includegraphics[width=8.2cm]{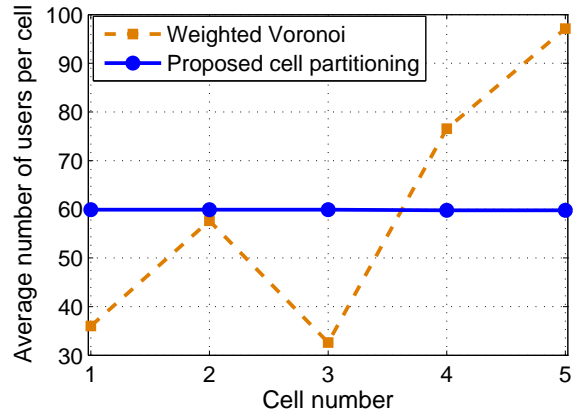}
		\vspace{-0.3cm}
		\caption{ Average number of users per cell partition.\vspace{-.32cm}}
		\label{NumbUsers}
	\end{center}\vspace{-.06cm}
\end{figure}


Fig.\,\ref{IntFactor} shows the average total data service as a function of the interference factor, $\beta$ used in (\ref{gamma}). From this figure we can see that, as the interference between UAVs decreases, the total data service that they can provide to the ground users increases. For instance, by decreasing $\beta$ from 1 (full interference case) to 0.1,  the total data service increases by a factor of 3 when 5 UAVs are deployed. Moreover, Fig.\,\ref{IntFactor} shows that the service gain achieved by using a higher number of UAVs is significant only when the interference between the UAVs is highly mitigated (low values of $\beta$). For example, increasing the number of UAVs from 5 to  10 can lead to 56\% data service gain for $\beta=0.1$, while this gain is only 5\% in the full interference case. Therefore, deploying more UAVs is beneficial in terms of data service if the interference between the UAVs is properly mitigated.

\begin{figure}[!t]
	\begin{center}
		\vspace{-0.1cm}
		\includegraphics[width=8.3cm]{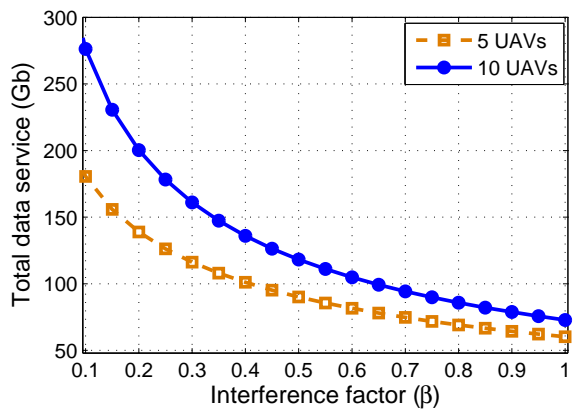}
		\vspace{-0.4cm}
		\caption{ Average data service versus interference factor.\vspace{-.32cm}}
		\label{IntFactor}
	\end{center}\vspace{-.3cm}
\end{figure}

In Fig.\,\ref{Load_Hover}, we show how the total data service changes as the maximum hover time of the UAVs increases. As expected, by increasing the hover time of each UAV, the users can be served for a longer time and, hence, the total data service increases. Fig.\,\ref{Load_Hover} also compares the performance of deploying 5 UAVs versus 10 UAVs. Interestingly, we can see that the 5-UAVs case with 40 minuets hover time (for each UAV) outperforms the 10-UAVs case with 30 minuets hover time. As a result, in this case, deploying UAVs that has a 33\% higher hover time is preferred than doubling the number of UAVs. In fact, increasing the number of UAVs  results in a higher interference which reduces the maximum data service gain that can be typically achieved by using more UAVs. Therefore, depending on system parameters, using more capable UAVs (i.e. with longer flight time) to service ground users can be more beneficial than deploying more UAVs with shorter flight times.\vspace{-0.2cm} 

\begin{figure}[!t]
	\begin{center}
		\vspace{-0.1cm}
		\includegraphics[width=8.3cm]{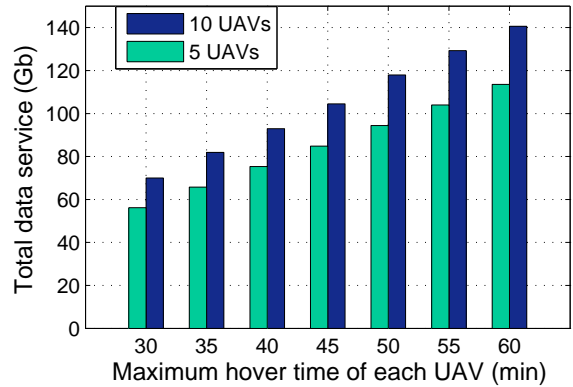}
		\vspace{-0.3cm}
		\caption{ Total data service versus the maximum hover time of each UAV.\vspace{-.32cm}}
		\label{Load_Hover}
	\end{center}\vspace{-.1cm}
\end{figure}

\subsection{Results for Scenario 2}

Here, we present the results for Scenario 2 in which the users are completely serviced using a minimum hover time. In this case, we consider a 10\,Mb data service requirement for each user.

\begin{figure}[!t]
	\begin{center}
		\vspace{-0.1cm}
		\includegraphics[width=8.3cm]{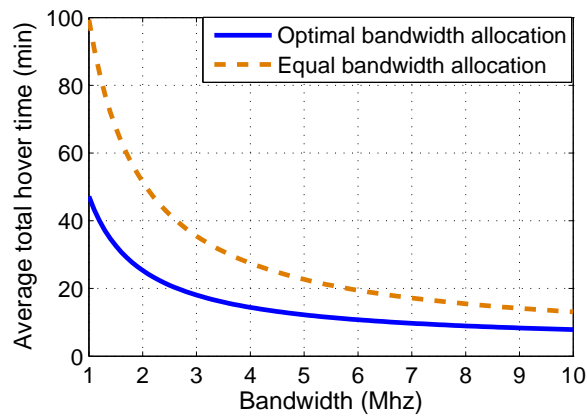}
		\vspace{-0.3cm}
		\caption{ Average hover time versus bandwidth.\vspace{-.32cm}}
		\label{Dwell_bw}
	\end{center}\vspace{-.1cm}
\end{figure} 

In Fig.\,\ref{Dwell_bw}, we show the total hover time versus the transmission bandwidth. Two bandwidth allocation schemes are considered, the optimal bandwidth allocation resulting from Proposition 2, and an equal bandwidth allocation. Clearly, by increasing the bandwidth, the total hover time required for serving the users decreases. In fact, a higher bandwidth can provide a higher the transmission rate and, hence, users can be serviced within a shorter time duration. From Fig.\,\ref{Dwell_bw}, we can see that, the optimal bandwidth allocation scheme can yield a 51\% hover time reduction compared to the equal bandwidth allocation. This is due the fact that, according to Proposition 2, by optimally assigning the bandwidth to each user based on its demand and location, the total hover time of UAVs can be minimized.

Fig.\,\ref{Hover_UAVs} shows the average total hover time of the UAVs as the number of UAVs varies. This result corresponds to the interference-free scenario in which the UAVs operate on different frequency bands. Hence, the total bandwidth usage linearly increases by increasing the number of UAVs. From Fig.\,\ref{Hover_UAVs}, we can see that the total hover time decreases as the number of UAVs increases. A higher number of UAVs corresponds to a higher number of cell partitions. Hence, the size of each cell partition decreases and the users will have a shorter distance to the UAVs. In addition, lower control time is required during serving a smaller and less congested cell. In fact, increasing the number of UAVs leads to a higher transmission rate, and lower control time thus leading to a lower hover time. For instance, as shown in Fig.\,\ref{Hover_UAVs}, when the number of UAVs increases from 2 ot 6, the total hover time decreases by 53\%. Nevertheless, deploying more UAVs in interference-free scenario results in a higher bandwidth usage. Therefore, there is a fundamental tradeoff between the hover time of UAVs and the bandwidth efficiency.  

\begin{figure}[!t]
	\begin{center}
		\vspace{-0.1cm}
		\includegraphics[width=8.3cm]{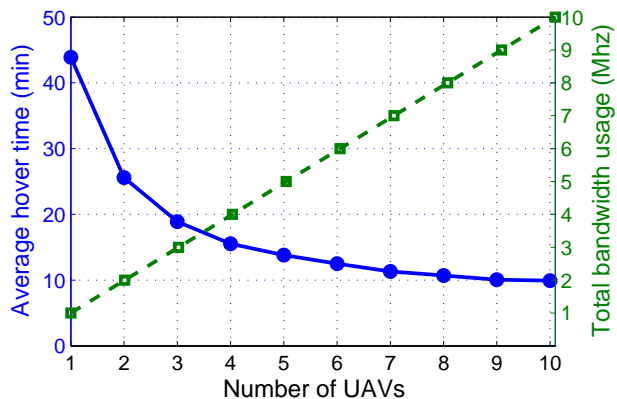}
		\vspace{-0.2cm}
		\caption{ Average hover time versus number of UAVs and bandwidth usage.\vspace{-.32cm}}
		\label{Hover_UAVs}
	\end{center}\vspace{-.1cm}
\end{figure}

Fig.\,\ref{Dwell_alpha} shows the impact of control time on the total hover time for the proposed cell partitioning, as a result of Theorem 3 and the weighted Voronoi diagram. In both cases, we use the optimal bandwidth allocation scheme. Clearly, as the control time factor, $\alpha$, increases, the total hover time also increases. From Fig.\,\ref{Dwell_alpha}, we can see that, using our proposed optimal cell partitioning approach, the average total hover time can be reduced by around 20\% compared to weighted Voronoi case. This is due to the fact that, unlike the weighted Voronoi, our approach also minimizes the control time while generating the cell partitions. We note that, the hover time difference between these two cases increases as $\alpha$ increases. In particular, as shown in Fig.\,\ref{Dwell_alpha}, our approach yields around 32\% hover time reduction when $\alpha=0.5$. 

\begin{figure}[!t]
	\begin{center}
		\vspace{-0.1cm}
		\includegraphics[width=8.3cm]{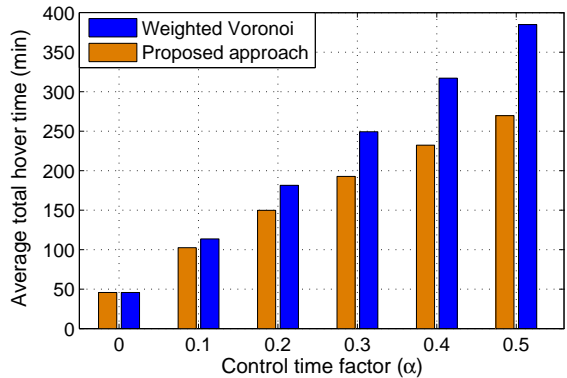}
		\vspace{-0.5cm}
		\caption{ Average hover time versus control time factor ($\alpha$) for $\sigma_o=200$\,m.\vspace{-.6cm}}
		\label{Dwell_alpha}
	\end{center}
\end{figure}

In Fig.\,\ref{HoverIntFactor}, we show the impact of interference on the hover time of UAVs. Clearly, the total hover time increases as the interference between the UAVs increases. This is due to the fact that a lower SINR leads to a lower transmission rate and, hence, a given UAV needs to hover for a longer time in order to completely service its users. For instance, the average hover time in the full interference case ($\beta=1$) is 4.5 times larger than the interference-free case in which $\beta=0$. Therefore, one can significantly reduce the hove time of UAVs by adopting interference mitigation techniques such as using orthogonal frequencies and scheduling of UAVs.

\begin{figure}[!t]
	\begin{center}
		\vspace{-0.1cm}
		\includegraphics[width=8.3cm]{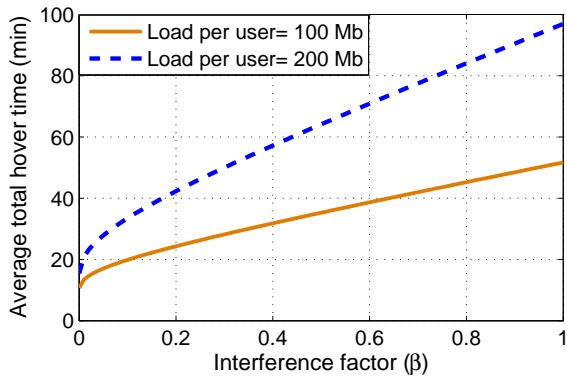}
		\vspace{-0.4cm}
		\caption{ Average hover time versus interference factor.\vspace{-.5cm}}
		\label{HoverIntFactor}
	\end{center}\vspace{.4cm}
\end{figure}


\section{conclusions}\vspace{-0.10cm}
In this paper, we have proposed a novel framework for optimizing UAV-enabled wireless networks while taking into account the flight time constraints of UAVs. In particular, we have investigated two UAV-based communication scenarios. First, given the maximum possible hover times of UAVs, we have maximized the average data service to the ground users under a fair resource allocation policy. To this end, using tools from optimal transport theory, we have determined the optimal cell partitions associated with the UAVs. In the second scenario, given the load requirements of users, we have minimized the average hover time of UAVs needed to completely serve the users. In this case, we have derived the optimal cell partitions as well as the optimal bandwidth allocation to the users that lead to the minimum hover time. The results have shown that, using our proposed cell partitioning approach, the users receive higher fair data service compared to the classical Voronoi case. Moreover, our results for the second scenario have revealed that the average hover time of UAVs can be significantly reduced by using our proposed approach. \vspace{-0.5cm}

\appendix \vspace{-0.2cm}
\subsection{Proof of Theorem 3}\vspace{-0.1cm}
	As shown in Proposition \ref{exist}, there exist optimal cell partitions $\mathcal{A}_i$, $i\in \mathcal{M}$ which are solutions to the optimization problem in (\ref{hover2}). Now, we consider two optimal partitions $\mathcal{A}_l$ and $\mathcal{A}_m$, and a point $\boldsymbol{v}_o=(x_o,y_o)\in \mathcal{A}_m$. Also, let $B_\epsilon(\boldsymbol{v}_o)$ be the intersection of $\mathcal{A}_m$ with a disk that has a center $v_o$ and radius $\epsilon >0$. To characterize the optimal solution of (\ref{hover2}), we first generate new cell partitions $\widehat {\mathcal{A}}_i$ (a variation of optimal partitions) as follows:\vspace{-0.15cm}
	\begin{equation}
	\left\{ \begin{array}{l}
	{\widehat {\mathcal{A}}_m} = \mathcal{A}_m\backslash {B_\varepsilon }({\boldsymbol{v}_o}),\\
	{\widehat {\mathcal{A}}_l} = \mathcal{A}_l \cup {B_\varepsilon }({\boldsymbol{v}_o}),\\
	{\widehat {\mathcal{A}}_i} = \mathcal{A}_i,\,\,\,\,i \ne l,m.
	\end{array} \right.
	\end{equation}
	
	Also, let ${a_\varepsilon } = \int_{{B_\varepsilon }({\boldsymbol{v}_o})} {f(x,y)\textrm{d}x\textrm{d}y} $, and $\widehat a_i= \int_{{\widehat{ \mathcal{A}}_i}} {f(x,y)\textrm{d}x\textrm{d}y}$. Considering the optimality of $\mathcal{A}_i$, $i\in \mathcal{M}$, we have:
		\begin{align}
		&\sum\limits_{i \in \mathcal{M}} \int_{{\mathcal{A}_i}} { {\left[\frac{{Nu(x,y)}}{{{C_i^{B_i}}(x,y)}}+q_i(x,y)\right]f(x,y)}\textrm{d}x\textrm{d}y} + g_i(a_i) \nonumber \\
		& {\mathop  \le \limits^{(a)} }\sum\limits_{i \in \mathcal{M}} \int_{{\widehat {\mathcal{A}}_i}} { {\left[\frac{{N u(x,y)}}{{{C_i^{B_i}}(x,y)}}+q_i(x,y)\right]f(x,y)}\textrm{d}x\textrm{d}y} + g_i(\widehat a_i), \nonumber\\
		&\int_{{\mathcal{A}_l}} { {\left[\frac{{Nu(x,y)}}{{{C_l^{B_l}}(x,y)}}+q_l(x,y)\right]f(x,y)}\textrm{d}x\textrm{d}y} + g_l(a_l)  + \int_{{\mathcal{A}_m}} { {\left[\frac{{Nu(x,y)}}{{{C_m^{B_m}}(x,y)}}+q_m(x,y)\right]f(x,y)}\textrm{d}x\textrm{d}y} + g_m(a_m)\nonumber\\
	&\le \int_{{\mathcal{A}_l}\cup {B_\varepsilon }({v_o})} { {\left[\frac{{Nu(x,y)}}{{{C_l^{B_l}}(x,y)}}+q_l(x,y)\right]f(x,y)}\textrm{d}x\textrm{d}y} + g_l(a_l+a_\epsilon) \nonumber \\&
	 + \int_{{\mathcal{A}_m \backslash {B_\varepsilon }({v_o})}} { {\left[\frac{{Nu(x,y)}}{{{C_m^{B_m}}(x,y)}}+q_m(x,y)\right]f(x,y)}\textrm{d}x\textrm{d}y} + g_m(a_m-a_\epsilon) \nonumber\\ 
	&\int_{{{B_\varepsilon }({v_o})}} { {\left[\frac{{Nu(x,y)}}{{{C_m^{B_m}}(x,y)}}+q_m(x,y)\right]f(x,y)}\textrm{d}x\textrm{d}y}+g_m(a_m)-g_m(a_m-a_\epsilon)\nonumber\\
	&\le \int_{{{B_\varepsilon }({v_o})}} { {\left[\frac{{Nu(x,y)}}{{{C_l^{B_l}}(x,y)}}+q_l(x,y)\right]f(x,y)}\textrm{d}x\textrm{d}y}+g_l(a_l+a_\epsilon)-g_l(a_l), \label{ineqa} 
		\end {align}
		where $(a)$ comes from the fact that $\mathcal{A}_i$ is optimal and, hence, any variation of that (${\widehat{\mathcal{A}}_i}$) cannot lead to a better solution.
		
		Now, we multiply both sides of the inequality in (\ref{ineqa}) by $\frac{1}{a_{\varepsilon}}$. Then, we take the limit when $\varepsilon \to 0$, and we use the following equality:
		\begin{align}
		&\mathop {\lim }\limits_{\varepsilon  \to 0}\frac{1}{a_\varepsilon} \int_{{{B_\varepsilon }({v_o})}}  {\left[\frac{{Nu(x,y)}}{{{C_l^{B_l}}(x,y)}}+q_l(x,y)\right]f(x,y)}\textrm{d}x\textrm{d}y=\mathop {\lim }\limits_{\varepsilon  \to 0}\frac{  \int_{{{B_\varepsilon }({v_o})}}{\left[\frac{{Nu(x,y)}}{{{C_l^{B_l}}(x,y)}}+q_l(x,y)\right]f(x,y)}  \textrm{d}x\textrm{d}y}{\int_{{{B_\varepsilon }({v_o})}}  f(x,y)\textrm{d}x\textrm{d}y}\nonumber\\
		 &=\frac{{Nu(x,y)}}{{{C_l^{B_l}}(x_o,y_o)}}+q_l(x_o,y_o).
		\end{align}
		 
		 Subsequently, following from (\ref{ineqa}), we have:
		 \begin{equation}\label{cond}
		  \frac{{Nu(x,y)}}{{{C_m^{B_m}}(x_o,y_o)}}+q_m(x_o,y_o)+g'_l(a_m)\le\frac{{Nu(x,y)}}{{{C_l^{B_l}}(x_o,y_o)}}+q_l(x_o,y_o)+g'_l(a_l).  
		 \end{equation}
Note that, (\ref{cond})	provides the condition under which a point $(x_o,y_o)$ is assigned to partition $m$ rather than $l$. Therefore, the optimal cell partitions can be characterized as:
	\begin{equation} \label{cellPar}
\hspace{-0.07cm} {\mathcal{A}_i^*} \hspace{-0.08cm} =\hspace{-0.08cm}  \left\{\hspace{-0.08cm}{(x,y)|\, \frac{Nu(x,y)}{C_i^{B_i}(x,y)}+q_i(x,y)+g'_i(a_i) \le \frac{Nu(x,y)}{C_j^{B_j}(x,y)}+q_j(x,y)+g'_j(a_j),\forall j \ne i \in \mathcal{M}} \right\}\hspace{-0.09cm}. \hspace{-0.5cm}
	\end{equation}
Finally, using (\ref{tau}), the optimal average hover time of UAV $i$ is:
\begin{equation}
{\tau_i^*} = \int_{{\mathcal{A}_i^*}} {\frac{{Nu(x,y)}}{{{C_i^{B_i}}(x,y)}}f(x,y)} \textrm{d}x\textrm{d}y + {g_i}\left( {\int_{{\mathcal{A}_i^*}} {f(x,y)} \textrm{d}x\textrm{d}y} \right),
\end{equation}
	
which proves the theorem.	\vspace{-0.2cm} 
\def\baselinestretch{1.08}
\bibliographystyle{IEEEtran}

\bibliography{references}

\begin{thebibliography}{10}
\providecommand{\url}[1]{#1}
\csname url@samestyle\endcsname
\providecommand{\newblock}{\relax}
\providecommand{\bibinfo}[2]{#2}
\providecommand{\BIBentrySTDinterwordspacing}{\spaceskip=0pt\relax}
\providecommand{\BIBentryALTinterwordstretchfactor}{4}
\providecommand{\BIBentryALTinterwordspacing}{\spaceskip=\fontdimen2\font plus
\BIBentryALTinterwordstretchfactor\fontdimen3\font minus
  \fontdimen4\font\relax}
\providecommand{\BIBforeignlanguage}[2]{{%
\expandafter\ifx\csname l@#1\endcsname\relax
\typeout{** WARNING: IEEEtran.bst: No hyphenation pattern has been}%
\typeout{** loaded for the language `#1'. Using the pattern for}%
\typeout{** the default language instead.}%
\else
\language=\csname l@#1\endcsname
\fi
#2}}
\providecommand{\BIBdecl}{\relax}
\BIBdecl

\bibitem{orfanus}
D.~Orfanus, E.~P. de~Freitas, and F.~Eliassen, ``Self-organization as a
  supporting paradigm for military {UAV} relay networks,'' \emph{IEEE
  Communications Letters}, vol.~20, no.~4, pp. 804--807, 2016.

\bibitem{mozaffari2}
M.~Mozaffari, W.~Saad, M.~Bennis, and M.~Debbah, ``Unmanned aerial vehicle with
  underlaid device-to-device communications: Performance and tradeoffs,''
  \emph{IEEE Transactions on Wireless Communications}, vol.~15, no.~6, pp.
  3949--3963, June 2016.

\bibitem{Ismail}
A.~Merwaday and I.~Guvenc, ``{UAV} assisted heterogeneous networks for public
  safety communications,'' in \emph{Proc. of IEEE Wireless Communications and
  Networking Conference Workshops (WCNCW)}, March 2015.

\bibitem{Zhang}
Y.~Zeng, R.~Zhang, and T.~J. Lim, ``Wireless communications with unmanned
  aerial vehicles: opportunities and challenges,'' \emph{IEEE Communications
  Magazine}, vol.~54, no.~5, pp. 36--42, May 2016.

\bibitem{IoTJournal}
M.~Mozaffari, W.~Saad, M.~Bennis, and M.~Debbah, ``Mobile unmanned aerial
  vehicles ({UAV}s) for energy-efficient {Internet of Things communications},''
  \emph{available online: arxiv.org/abs/1703.05401}, 2017.

\bibitem{HouraniModeling}
A.~Hourani, S.~Kandeepan, and A.~Jamalipour, ``Modeling air-to-ground path loss
  for low altitude platforms in urban environments,'' in \emph{Proc. of IEEE
  Global Communications Conference (GLOBECOM)}, Austin, TX, USA, Dec. 2014.

\bibitem{Azari}
M.~M. Azari, F.~Rosas, K.~C. Chen, and S.~Pollin, ``Joint sum-rate and power
  gain analysis of an aerial base station,'' in \emph{Proc. of IEEE Global
  Communications Conference (GLOBECOM) Workshops}, Dec. 2016.

\bibitem{Kalantari}
E.~Kalantari, H.~Yanikomeroglu, and A.~Yongacoglu, ``On the number and {3D}
  placement of drone base stations in wireless cellular networks,'' in
  \emph{Proc. of IEEE Vehicular Technology Conference}, Sep. 2016.

\bibitem{bor}
I.~Bor-Yaliniz and H.~Yanikomeroglu, ``The new frontier in ran heterogeneity:
  Multi-tier drone-cells,'' \emph{IEEE Communications Magazine}, vol.~54,
  no.~11, pp. 48--55, 2016.

\bibitem{Letter}
M.~Mozaffari, W.~Saad, M.~Bennis, and M.~Debbah, ``Efficient deployment of
  multiple unmanned aerial vehicles for optimal wireless coverage,'' \emph{IEEE
  Communications Letters}, vol.~20, no.~8, pp. 1647--1650, Aug. 2016.

\bibitem{Sky}
Facebook, ``Connecting the world from the sky,'' Facebook Technical Report,
  2014.

\bibitem{Jeong}
S.~Jeong, O.~Simeone, and J.~Kang, ``Mobile edge computing via a {UAV}-mounted
  cloudlet: Optimal bit allocation and path planning,'' \emph{available online:
  https://arxiv.org/abs/1609.05362.}, 2016.

\bibitem{jiang}
F.~Jiang and A.~L. Swindlehurst, ``Optimization of {UAV} heading for the
  ground-to-air uplink,'' \emph{IEEE Journal on Selected Areas in
  Communications}, vol.~30, no.~5, pp. 993--1005, June 2012.

\bibitem{ZhangEnergy}
Y.~Zeng and R.~Zhang, ``Energy-efficient {UAV} communication with trajectory
  optimization,'' \emph{IEEE Transactions on Wireless Communications}, to
  appear 2017.

\bibitem{Qin}
Q.~Wu, Y.~Zeng, and R.~Zhang, ``Joint trajectory and communication design for
  {UAV}-enabled multiple access,'' \emph{available online:
  https://arxiv.org/abs/1704.01765}, 2017.

\bibitem{Vishnu}
V.~V. Chetlur and H.~S. Dhillon, ``Downlink coverage analysis for a finite {3D}
  wireless network of unmanned aerial vehicles,'' \emph{available online:
  arxiv.org/abs/1701.01212}, 2017.

\bibitem{Vishal}
V.~Sharma, M.~Bennis, and R.~Kumar, ``{UAV}-assisted heterogeneous networks for
  capacity enhancement,'' \emph{IEEE Communications Letters}, vol.~20, no.~6,
  pp. 1207--1210, June 2016.

\bibitem{OTUAV}
M.~Mozaffari, W.~Saad, M.~Bennis, and M.~Debbah, ``Optimal transport theory for
  power-efficient deployment of unmanned aerial vehicles,'' in \emph{Proc. of
  IEEE International Conference on Communications (ICC)}, May 2016.

\bibitem{niu}
S.~Niu, J.~Zhang, F.~Zhang, and H.~Li, ``A method of {UAVs} route optimization
  based on the structure of the highway network,'' \emph{International Journal
  of Distributed Sensor Networks}, 2015.

\bibitem{DroneDel}
K.~Dorling, J.~Heinrichs, G.~G. Messier, and S.~Magierowski, ``Vehicle routing
  problems for drone delivery,'' \emph{IEEE Transactions on Systems, Man, and
  Cybernetics: Systems}, vol.~47, no.~1, pp. 70--85, Jan 2017.

\bibitem{AkramMagazin}
S.~Chandrasekharan, K.~Gomez, A.~Al-Hourani, S.~Kandeepan, T.~Rasheed,
  L.~Goratti, L.~Reynaud, D.~Grace, I.~Bucaille, T.~Wirth, and S.~Allsopp,
  ``Designing and implementing future aerial communication networks,''
  \emph{IEEE Communications Magazine}, vol.~54, no.~5, pp. 26--34, May 2016.

\bibitem{villani}
C.~Villani, \emph{Topics in optimal transportation}.\hskip 1em plus 0.5em minus
  0.4em\relax American Mathematical Soc., 2003, no.~58.

\bibitem{FDMA2}
O.~Lysenko, S.~Valuiskyi, P.~Kirchu, and A.~Romaniuk, ``Optimal control of
  telecommunication aeroplatform in the area of emergency,'' \emph{Information
  and Telecommunication Sciences}, no.~1, 2013.

\bibitem{ITUR}
ITU-R, ``Rec. p.1410-2 propagation data and prediction methods for the design
  of terrestrial broadband millimetric radio access systems,'' \emph{Series,
  Radiowave propagation}, 2003.

\bibitem{VoronoiDiag}
F.~Aurenhammer, ``Voronoi diagrams—a survey of a fundamental geometric data
  structure,'' \emph{ACM Computing Surveys (CSUR)}, vol.~23, no.~3, pp.
  345--405, 1991.

\bibitem{Alonso}
A.~Silva, H.~Tembine, E.~Altman, and M.~Debbah, ``Optimum and equilibrium in
  assignment problems with congestion: Mobile terminals association to base
  stations,'' \emph{IEEE Transactions on Automatic Control}, vol.~58, no.~8,
  pp. 2018--2031, Aug. 2013.

\bibitem{Flippo}
F.~Santambrogio, ``Optimal transport for applied mathematicians,''
  \emph{Birk{\"a}user, NY}, 2015.

\bibitem{ambros}
L.~Ambrosio and N.~Gigli, ``A user’s guide to optimal transport,'' in
  \emph{Modelling and optimisation of flows on networks}.\hskip 1em plus 0.5em
  minus 0.4em\relax Springer, 2013, pp. 1--155.

\bibitem{Crippa}
G.~Crippa, C.~Jimenez, and A.~Pratelli, ``Optimum and equilibrium in a
  transport problem with queue penalization effect,'' \emph{Advances in
  Calculus of Variations}, vol.~2, no.~3, pp. 207--246, 2009.

\bibitem{ghazzai}
H.~Ghazzai, ``Environment aware cellular networks,'' \emph{available online:
  http://repository.kaust.edu.sa/kaust/handle/10754/344436}, Feb. 2015.

\bibitem{Quadcopter}
M.~C. Achtelik, J.~Stumpf, D.~Gurdan, and K.~M. Doth, ``Design of a flexible
  high performance quadcopter platform breaking the {MAV} endurance record with
  laser power beaming,'' in \emph{Proc. of IEEE International Conference on
  Intelligent Robots and Systems}, Sep. 2011.

\bibitem{jain}
R.~Jain, D.-M. Chiu, and W.~R. Hawe, \emph{A quantitative measure of fairness
  and discrimination for resource allocation in shared computer system}.\hskip
  1em plus 0.5em minus 0.4em\relax tech. rep., Digital Equipment Corporation,
  DEC-TR-301, 1984, vol.~38.

\end{thebibliography}
\end{document}